\documentclass[a4paper,onecolumn,11pt,accepted=2024-05-09]{quantumarticle}

\pdfoutput=1
\usepackage[utf8]{inputenc}
\usepackage[english]{babel}
\usepackage[T1]{fontenc}
\usepackage{hyperref}
\usepackage[numbers,sort&compress]{natbib}
\usepackage{qcircuit}
\usepackage{graphics}
\usepackage{amsmath,amssymb,amsthm,mathrsfs,amsfonts,dsfont}
\usepackage{subfigure, epsfig}
\usepackage{braket}
\usepackage{bm}
\usepackage{enumerate}
\usepackage{algorithm}
\usepackage{physics}
\usepackage{pgfplots}
\pgfplotsset{width=7cm,compat=1.9}
\usepgfplotslibrary{external}
\usepackage{orcidlink}

\usepackage{algpseudocode}
\usepackage{color}

\graphicspath{{./figure/}}

\newtheorem{theorem}{Theorem}
\newtheorem{lemma}{Lemma}
\newtheorem{corollary}{Corollary}

\newtheorem{observation}{Observation}
\newtheorem{definition}{Definition}
\newtheorem{proposition}{Proposition}

\newcommand{\renyi}{R$\mathrm{\acute{e}}$nyi }
\newcommand{\renyia}{R$\mathrm{\acute{e}}$nyi-$\alpha$ }
\newcommand{\renyiaa}{S$\textrm{R}_{\alpha}$E}

\newcommand{\renyih}{S$\textrm{R}_{\frac{1}{2}}$E}
\newcommand{\renyit}{S$\textrm{R}_{2}$E}

\newcommand{\kb}[1]{\ket{#1}\bra{#1}}
\newcommand{\id}{\mathbb{I}}

\newcommand{\comments}[1]{}

\newcommand{\mc}{\mathcal}
\newcommand{\mb}{\mathbf}
\newcommand{\mbb}{\mathbb}
\newcommand{\mr}{\mathrm}
\begin{document}

\title{Magic of quantum hypergraph states}


\author{Junjie Chen}
\affiliation{Center for Quantum Information, Institute for Interdisciplinary Information Sciences, Tsinghua University, Beijing 100084, China}
\author{Yuxuan Yan}
\affiliation{Center for Quantum Information, Institute for Interdisciplinary Information Sciences, Tsinghua University, Beijing 100084, China}
\author{You Zhou~\orcidlink{0000-0003-0886-077X}}
\email{you\_zhou@fudan.edu.cn}
\affiliation{Key Laboratory for Information Science of Electromagnetic Waves (Ministry of Education), Fudan University, Shanghai 200433, China}
\affiliation{Hefei National Laboratory, Hefei 230088, China}

\begin{abstract}
Magic, or nonstabilizerness, characterizes the deviation of a quantum state from the set of stabilizer states, playing a fundamental role in quantum state complexity and universal fault-tolerant quantum computing. However, analytical and numerical characterizations of magic are very challenging, especially for multi-qubit systems, even with a moderate qubit number. Here, we systemically and analytically investigate the magic resource of archetypal multipartite quantum states---quantum hypergraph states, which can be generated by multi-qubit controlled-phase gates encoded by hypergraphs. We first derive the magic formula in terms of the stabilizer R$\mathrm{\acute{e}}$nyi-$\alpha$ entropies for general quantum hypergraph states. If the average degree of the corresponding hypergraph is constant, we show that magic cannot reach the maximal value, i.e., the number of qubits $n$. Then, we investigate the statistical behaviors of random hypergraph states' magic and prove a concentration result, indicating that random hypergraph states typically reach the maximum magic. This also suggests an efficient way to generate maximal magic states with random diagonal circuits. Finally, we study hypergraph states with permutation symmetry, such as $3$-complete hypergraph states, where any three vertices are connected by a hyperedge. Counterintuitively, such states can only possess constant or even exponentially small magic for $\alpha\geq 2$. Our study advances the understanding of multipartite quantum magic and could lead to applications in quantum computing and quantum many-body physics.
\end{abstract}

\maketitle

\section{Introduction}

Stabilizer formalism and Clifford group are fundamental building blocks of quantum information science, especially for quantum error-correcting codes \cite{Terhal2015QECC}. Quantum states beyond stabilizer formalism own nonstabilizerness or so-called ``magic'' \cite{Bravyi2005Magic,veitch2014resource}, a resource that enables universal fault-tolerant quantum computing \cite{campbell2017roads} via the magic-state-injection approach \cite{Bravyi2005Magic,gottesman1999demonstrating}. Meanwhile, magic also characterizes complexity beyond entanglement \cite{Horodecki2007Entanglement}. For example, Clifford circuits, creating volume-law entangled states, are classically simulatable via Gottesman-Knill theorem \cite{knill2005quantum}. With the additional magic content, the classical simulation cost scales exponentially \cite{Aaronson2004ImprovedGKT}, and faster simulation algorithms were proposed by quantifying magic more exquisitely \cite{Bravyi2016Trading,Bravyi2016Improved,Bravyi2019simulation,Pashayan2015neg,Howard2017Resource,Seddon2021Speedups}. Along this line, there is an increasing interest in the role of magic in quantum many-body physics, as a significant supplement for entanglement \cite{Amico2008Entanglement}, with topics ranging from phases of quantum matter, e.g., topological order \cite{Sarkar2020NJP,ellison2021symmetry,Liu2022manybodymagic} to non-equilibrium dynamics \cite{Shiyu2020Entanglemagic, haferkamp2020quantum,true2022transitions,Sewell2022Manathermal}, e.g., quantum chaos \cite{Leone2021Chaos,goto2022probing,garcia2023resource} and black-hole physics \cite{White2021CFTmagic,leone2022learning, leone2022retrieving}. As such, investigating quantum magic can fertilize quantum computing and deepen our understanding of quantum complexity.

Quantum magic has been extensively studied in the resource framework with various measures proposed 
\cite{veitch2012negative,Seddon2019channel,Beverland2020NJP,wang2019quantifying,wang2020efficiently,Haug2022Bellmagic,saxena2022quantifying,bu2023quantum}. Most of these measures involve optimization in the stabilizer polytope, which is extremely challenging for multi-qubit systems analytically and numerically \cite{Bravyi2016Trading,Howard2017Resource,Heinrich2019robustnessofmagic}. To address this issue, Stabilizer \renyi Entropy (SRE) was recently proposed as a faithful indicator of magic \cite{Leone2022SRE,leone2024stabilizer}, which explicitly quantifies magic by the weight distribution of the state projected to all Pauli strings. The simplicity of SRE triggers a series of interesting studies, whose topics include magic measurement protocol \cite{oliviero2022measuring, Haug2022Bellmagic,haug2023efficient}, the complexity of wave functions \cite{oliviero2022magic,odavic2022complexity}, quantum information dynamics \cite{rattacaso2023stabilizer,piemontese2023entanglement}, and learning theory \cite{leone2023learningTdope,leone2023nonstabilizerness}. 
Even though SREs enable an explicit computation and experimental realization, the evaluation cost scales exponentially with the number of qubits, in general, \cite{oliviero2022measuring}, thereby hindering its application towards large-scale systems. There are some positive progresses using matrix-product-state \cite{Haug2022MPSmagic,lami2023quantum,haug2023stabilizer,tarabunga2023many}. However, the target states are limited to low-entanglement ones, and the methods are mainly for numerical purposes. Till now, SRE and magic have been unexplored for complex many-qubit states with extensive entanglement. 

In this work, we significantly extend the scope of magic quantification to large-scale and highly entangled states. In particular, we systemically and analytically investigate the magic of quantum hypergraph states \cite{Rossi2013Hyper,Qu2013hypergraph}, which are generalized from graph states \cite{Raussendorf2001onewayQC,Raussendorf2003MBQC}. Unlike graph states, which are generated by Clifford gates and lack quantum magic, hypergraph states are capable of magic and play an essential role in quantum advantage protocols \cite{Bremner2016IQP}, measurement-based quantum computing (with Pauli measurements) \cite{Miller2016Hierarchy,Takeuchi2019PauliMBQC} and topological order \cite{Levin2012Braiding,Yoshida2016Topological,Miller2018Latent}. According to the indices of all Pauli strings, we relate the magic in terms of SRE to a family of induced hypergraphs from the original one. This pictorial expression enables a series of analytical findings as follows. We first show a general upper bound of magic for any hypergraph state with a bounded average degree, for instance, ones whose hypergraphs are defined on lattices like Union-Jack one \cite{Miller2018Latent}. We further develop general theories that transform the statistical properties of magic into a series of counting problems in the binary domain. Our theories lead to the concentration result that the magic of hypergraph states is typically large and very near the maximal value, showing similar behavior to the unphysical Haar random states \cite{Leone2022SRE,Liu2022manybodymagic,white2020mana}. In addition, we analyze the magic of quantum hypergraph states with permutation symmetry. Based on the symmetry simplification and pictorial derivation, we obtain exact analytical results of the stabilizer \renyia entropy (\renyiaa) for different $\alpha$'s, and in particular, find that S$\textrm{R}_{2}$E and S$\textrm{R}_{\frac{1}{2}}$E can be exponentially different for these states. Specifically, S$\textrm{R}_{\frac{1}{2}}$E serves as a lower bound for the robustness of magic, which operationally quantifies the overhead in classical simulation when using ancillary magic states \cite{Howard2017Resource, Leone2022SRE}. Our findings and the developed techniques can advance further investigations of multipartite quantum magic with applications from quantum computing to quantum many-body physics, where especially hypergraph states can serve as an archetypal class of complex states and tractable toy models of other complex quantum systems.

\section{Preliminaries}

\subsection{Quantum hypergraph states}

Graph states are widely recognized for their well-defined structure and robust entanglement capabilities, making them popular in various research fields such as quantum entanglement \cite{Hein2004Multiparty}, quantum computing \cite{Raussendorf2001onewayQC,Raussendorf2003MBQC}, and quantum error correction \cite{bell2014experimental}. Despite their numerous applications, graph states lack quantum magic and are unsuitable for demonstrating quantum advantage in certain scenarios. To overcome this limitation, the extension from graph states to hypergraph states has been proposed \cite{Rossi2013Hyper}. This generalization endows them with the necessary quantum magic while retaining their clear structural advantage, expanding the applicability in quantum computing \cite{Bremner2016IQP,Miller2016Hierarchy}.

Here we give a brief introduction to quantum hypergraph states \cite{Kruszynska2009Hyper,Rossi2013Hyper,Qu2013hypergraph}, which is a generalization of graph states. A hypergraph $G=(V,E)$ consists of a vertex set $V=\{v_i|i\in[n]\}$ and a hyperedge set $E=\{e_j|j\in[m],e_j\subseteq V,e_j\neq\emptyset\}$. Here $[n]=\{1,2,\cdots,n\}$ and $[m]=\{1,2,\cdots,m\}$ are the index sets of vertices and hyperedges, respectively. A hyperedge $e$ can connect $c\geq 1$ vertices, denoted as $|e|=c$, and we call it a $c$-edge. There are totally $\sum_{c=1}^n \binom{n}{c}=2^n-1$ possible hyperedges. If $\forall e\in E$, $|e|=c$, the hypergraph is called $c$-uniform hypergraph. And a $c$-complete hypergraph contains all such $c$-edges. See Fig.~\ref{fig:inducedG} (a) for an instance of a hypergraph with six vertices and four hyperedges.

\begin{figure}[tbhp!]
\centering
\resizebox{8.5cm}{!}{\includegraphics[scale=0.8]{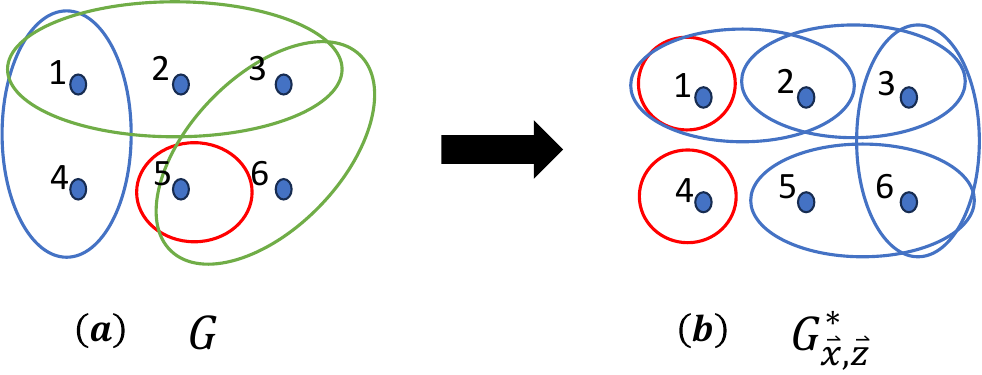}}
\caption{(a) A hypergraph with six vertices. We use big circles to label hyperedges and small points to label vertices. There are four hyperedges, $\{1,2,3\},\{3,5,6\},\{1,4\},\{5\}$, and the first two of them are $3$-edge.  (b) shows the induced hypergraph $G^*_{\vec{x},\vec{z}}$ according to Eq.~\eqref{eq:Gin} and \eqref{eq:GinSim}. Here two 6-bit strings are chosen as $\vec{x}=\{1,0,1,0,1,0\}$ and $\vec{z}=\{1,0,0,1,0,0\}$. For example, the 2-edge $\{1,2\}$ is induced from the $3$-edge $\{1,2,3\}$ of $G$ in (a) with $x_3=1$, according to the edge set $E^{(2)}_{\vec{x}}$ in Eq.~\eqref{eq:Gin}; the 1-edge $\{4\}$ is directly by $z_4=1$ according to $E^{(1)}_{\vec{z}}$ in Eq.~\eqref{eq:GinSim}.}  \label{fig:inducedG} 
\end{figure}

We use the label $e\setminus e'$ to denote the minus operation of two edges with $e'\subset e$, and specifically $e \setminus \{v\}$ is a new edge without the vertex $v$. 
The neighbor set of a vertex $v_i$ contains all the vertices $v_j$ connected to $v_i$ by some hyperedge, and the degree $\Delta(v_i)$ is the cardinality of this set. The average degree of a hypergraph $G$ is defined as $\bar{\Delta}(G)=n^{-1}\sum_{i=1}^n\Delta(v_i)$.

Associate each vertex $v_i$ with a qubit, i.e., a local Hilbert space $\mathcal H_i\simeq\mathbb{C}^2$, and the total Hilbert space is $\mathcal H=\otimes_{i=1}^n \mathcal {H}_i$ with the total dimension $d=2^n$. Apply quantum gates on the qubits associated with hyperedges, and one can define the quantum hypergraph state as follows. 
\begin{definition}
    Given a hypergraph $G=(V,E)$ with $n$ vertices, the corresponding quantum hypergraph state of $n$ qubits reads
    \begin{equation}\label{eq:def:hyper}
	\ket{G}:=U(G)\ket{+}^{\otimes n}=\prod_{e\in E}CZ_e\ket{+}^{\otimes n},
    \end{equation}
    where $\ket{+}=\frac1{\sqrt{2}}(\ket{0}+\ket{1})$, the phase unitary $U(G)$ is completely determined by the hypergraph $G$, and $CZ_e=\bigotimes_{v_i\in e}\mathbb{I}_i-2\bigotimes_{v_i\in e}\ket{1}_i\bra{1}$ the generalized Controlled-$Z$ gate acting non-trivially on the support of edge $e$.
\end{definition}

Quantum hypergraph states are generally not traditional stabilizer states since $CZ_e$ are not Clifford gates as $|e|>2$. Fortunately, one can still apply the generalized stabilizer formalism as follows. 
\begin{definition} 
    For a hypergraph state $\ket{G}$ defined in Eq.~\eqref{eq:def:hyper}, it is uniquely determined by the following $n$ independent (generalized) stabilizer generators
    \begin{equation}\label{eq:stabG}
	S_i=X_i\prod_{e\in E, e\ni v_i} CZ_{e\setminus \{v_i\}},\quad i\in[n]
    \end{equation}
    such that $S_i\ket{G}=\ket{G}$.
\end{definition}
Note that here $S_i$ is not necessarily in the tensor-product of single-qubit Pauli operator $\{X_i,Y_i,Z_i\}$ and the identity $\id_i$, but the eigenvalues of $S_i$ are still $\pm 1$. Similar to the graph states and stabilizer states \cite{toth2005detecting}, hypergraph states can be written as the successive projection to the $+1$ space of each generator $S_i$ as 
\begin{equation}\label{eq:stabDecom}
    \kb{G}=\prod_{i=1}^n \frac{\id+S_i}{2}=2^{-n}\sum_{\vec{s}}\prod_{i}S_i^{s_i},
\end{equation}
where $\vec{s}$ is a binary vector of dimension $n$ and the summation is over all possible $\vec{s}$. It is direct to see that any multiplication $\prod_{i}S_i^{s_i}$ is also a stabilizer of $\ket{G}$. 

We show the following observation of a very compact form for any stabilizer of $\ket{G}$, i.e., the multiplication of some Pauli X gates and some phase gates determined by the hypergraph structure. For consistency of the following discussion, we let $CZ_{\emptyset}=-1$.
\begin{observation}\label{ob:beauty}
For an $n$-qubit quantum hypergraph state $\ket{G}$ with the stabilizer generator $S_i$ defined in Eq.~\eqref{eq:stabG}, any stabilizer $\mathrm{St}(G,\vec{s})$ labeled by the vector $\vec{s}$ shows
\begin{equation}\label{eq:beauty}
\begin{aligned}
    \mathrm{St}(G,\vec{s})&:=\prod_{i}S_i^{s_i}
    =\prod_{i}X_i^{s_i}\cdot\prod_{e\in E}\prod_{q\neq \emptyset,q\subseteq v(\vec{s})\cap e}CZ_{e\setminus q},
\end{aligned}
\end{equation}
where $v(\vec{s})$ is the set of vertices $v_i$ with the corresponding $s_i=1$. 
\end{observation}
Observation \ref{ob:beauty} can be proved by recursively using the commuting relation that $X_iCZ_eX_i=CZ_e\cdot CZ_{e\setminus\{v_i\}}$ when $v_i\in e$ and $X_iCZ_eX_i=CZ_e$ when $v_i\notin e$ \cite{guhne2014hypergraph}, which is left in Appendix \ref{App:beauty}. This way, one can finally move the Pauli X gates and phase gates apart. Observation \ref{ob:beauty} could be of independent interest and applied to other studies of hypergraph states. 
Eq.~\eqref{eq:beauty} is very helpful for the following discussions, and for the simplicity of the presentation, hereafter, we always take $q\neq \emptyset$ and $e\in E$ in the product by default. 

\subsection{Quantum magic and stabilizer R\'enyi entropy}
Magic \cite{Bravyi2005Magic,veitch2014resource} quantifies the derivation of a quantum state from the stabilizer states, which is an essential resource for quantum computing complexity and its fault-tolerant realization \cite{Aaronson2004ImprovedGKT,campbell2017roads}. 
Stabilizer R\'enyi entropy (SRE) \cite{Leone2022SRE,oliviero2022measuring} was a recently introduced faithful indicator of magic for (pure) multipartite states defined via the probability distribution from the projection onto the Pauli operators as follows.

A quantum state $\rho$ can be decomposed onto the complete Pauli operator basis, i.e., Pauli-Liouville representation,
\begin{equation}
    \rho=\sum_{P\in\mathcal{P}_n}2^{-n}\Tr{P\rho}P.
\end{equation}
Here we consider $n$-qubit system and $\mathcal{P}_n$ is the Pauli group $\{\id_i,X_i,Y_i,Z_i\}^{\otimes n}$ ignoring the phase, which can be denoted as 
\begin{equation}\label{eq:pauliP}
\begin{aligned}
    P_{\vec{x},\vec{z}}
    =\omega(\vec{x},\vec{z}) X^{\vec{x}} Z^{\vec{z}}.
    \end{aligned}
\end{equation}
Here $\vec{x}$ and $\vec{z}$ are binary vectors of dimension $n$, $X^{\vec{x}}=\bigotimes_{i=1}^n X_i^{x_i}$ and similar for $Z^{\vec{z}}$, and $\omega(\vec{x},\vec{z})=\sqrt{-1}^{\sum_{i=1}^n x_iz_i}$ is some unessential phase.

For a pure state $\rho=\kb{\Psi}$, one can utilize the orthogonality of the Pauli operator basis to express the purity as
\begin{equation}
\begin{aligned}
    \Tr{\rho^2}&=\sum_{P_i,P_j\in\mathcal{P}_n}2^{-2n}\Tr{P_i\rho}\Tr{P_j\rho}\Tr{P_iP_j}\\
    &=\sum_{P_i\in\mathcal{P}_n}2^{-n}\Tr{P_i\rho}^2=1.
\end{aligned}
\end{equation}

In this way, the non-negative terms in the summation of the second line can be regarded as a probability distribution, and S$\textrm{R}_{\alpha}$E of $\ket{\Psi}$ is defined via the \renyia entropy of this distribution.
\begin{equation}
    \mathbf{M}_{\alpha}(\ket{\Psi})=\frac{1}{1-\alpha}\log{\sum_{P\in\mathcal{P}_n}\left(2^{-n}\Tr{P\ket{\Psi}\bra{\Psi}}^2\right)^{\alpha}}-n,
\end{equation}
where the offset $-n$ keeps the magic of stabilizer states to be zero. Hereafter, all the $\log$ functions are base two otherwise specified.
\comments{
Specifically, for $\alpha=2$, S$R_2$E shows
\begin{equation}
\begin{aligned}
    \mathbf{M}_{2}(\ket{\Psi})=-\log{\sum_{P\in\mathcal{P}_n}2^{-n}\left(\Tr{P\ket{\Psi}\bra{\Psi}}\right)^{4}}.
\end{aligned}
\end{equation}
}

As an entropy function defined on the domain $\mathcal{P}_n$ with $d^2$ elements, a direct upper bound shows $\mathbf{M}_{\alpha}(\ket{\Psi})\leq n$. Note also that $\mathbf{M}_{\alpha}$ monotonically decreases with $\alpha$ by the property of \renyi entropy, and thus $\mathbf{M}_{2}$ serves as a lower bound for $\mathbf{M}_{\alpha\leq 2}$. It is shown that its average on Haar random states, $\langle \mathbf{M}_{2}(\ket{\Psi})\rangle_{\Psi\in \mathrm{Haar}} \geq \log_2(2^n+3)-2>n-2$ \cite{Leone2022SRE}, very close to the maximum possible value. We also remark that \renyih\ severs as a lower bound of another important measure, robustness of magic, i.e., $\log(\mathbf{R}(\ket{\Psi}))\geq 1/2\mathbf{M}_{\frac1{2}}(\ket{\Psi})$ \cite{Howard2017Resource}, which could quantify the cost of classical simulation.

For ease of the following discussion, we define the closely related quantity $\alpha$-order \emph{Pauli-Liouville(PL) moment}
as  
\begin{equation}\label{eq:malpha}
    \mathbf{m}_{\alpha}(\ket{\Psi})=2^{-n}\sum_{P\in\mathcal{P}_n}\left(\Tr{P\ket{\Psi}\bra{\Psi}}\right)^{2\alpha},
\end{equation}
and the corresponding \renyiaa\ directly reads 
\begin{equation}\label{eq:SREalpha}
   \mathbf{M}_{\alpha}(\ket{\Psi})=(1-\alpha)^{-1}\log {\mb{m}_{\alpha}(\ket{\Psi})}.
\end{equation}

\section{Stabilizer R\'enyi Entropy of hypergraph states}\label{sec:SREG}
In this section, we show a general formula of the magic for any hypergraph state by relating the PL-moment and, thus, SRE to a family of induced hypergraphs. This pictorial result enables us to find a general upper bound of magic based on the structure of the corresponding hypergraph, which constrains the magic, especially for the hypergraph states on the lattice.

First, let us define a family of hypergraphs $G_{\vec{x},\vec{z}}=(V,E_{\vec{x},\vec{z}})$, which are induced from the original hypergraph $G$. The vertex set $V$ remains the same, and the updated edge set $E_{\vec{x},\vec{z}}$ is determined by two $n$-bit vectors $\vec{x}$ and $\vec{z}$ shown as follows. Hereafter, all the additions are module 2 on the binary domain otherwise specified.
\begin{equation}\label{eq:Gin}
\begin{aligned}
    E_{\vec{x},\vec{z}}=&E^{(1)}_{\vec{x},\vec{z}}\cup E^{(2)}_{\vec{x}},\\
    E^{(1)}_{\vec{x},\vec{z}}=&\left\{e_1=\{v_j\} \Bigg| z_j+\sum_{e\ni v_j}\prod_{i:v_i\in e\setminus\{v_j\}}x_i=1 \right\},\\
    E^{(2)}_{\vec{x}}=&\left\{e_2\subseteq V \Bigg| |e_2|\geq 2,\sum_{e\supset e_2}\prod_{i:v_i\in e\setminus e_2}x_i=1 \right\}.
\end{aligned}
\end{equation}
Here $E^{(1)}_{\vec{x},\vec{z}}\cap E^{(2)}_{\vec{x}}=\emptyset$, and $E^{(1)}_{\vec{x},\vec{z}}$ denotes the $1$-edge set, while $E^{(2)}_{\vec{x}}$ is for the set with $2$ or more cardinality edges. Note that $E^{(2)}_{\vec{x}}$ only depends on $\vec{x}$ , and is irrelevant to $\vec{z}$. Following the definition in Eq.~\eqref{eq:def:hyper}, we denote the phase unitary encoded by this hypergraph $G_{\vec{x},\vec{z}}$ as $U(G_{\vec{x},\vec{z}})$.

Additionally, we define another induced hypergraph by simplifying the 1-edge set of $G_{\vec{x},\vec{z}}$.
\begin{equation}\label{eq:GinSim}
\begin{aligned}
    &G^*_{\vec{x},\vec{z}}=(V,E^{*}_{\vec{x},\vec{z}}),\\
    &E^{*}_{\vec{x},\vec{z}}=E^{(1)}_{\vec{z}}\cup E^{(2)}_{\vec{x}},\ E^{(1)}_{\vec{z}}=\{e_1=\{v_j\} | z_j=1 \},
\end{aligned}
\end{equation}
where $E^{(1)}_{\vec{z}}$ is the simplified 1-edge set, only determined by the index $\vec{z}$ compared to $E^{(1)}_{\vec{x},\vec{z}}$ in Eq.~\eqref{eq:Gin}. See Fig.~\ref{fig:inducedG} (b) for an illustration of an induced hypergraph. The corresponding phase unitary is denoted as $U(G^*_{\vec{x},\vec{z}})$. 

The following result relates the PL-component of any hypergraph state to the induced phase unitary. 
\begin{proposition}\label{prop:PLcomp}
Given a hypergraph state $\ket{G}$, the square of the PL-component respective to the Pauli operator $P_{\vec{x},\vec{z}}$ shows
\begin{equation}\label{eq:PLcomp}
\begin{aligned}
    \Tr{P_{\vec{x},\vec{z}}\kb{G}}^2=2^{-2n}\Tr{U(G_{\vec{x},\vec{z}})}^2,
\end{aligned}
\end{equation}
where $U(G_{\vec{x},\vec{z}})$ is the phase unitary determined by the hypergraph $G_{\vec{x},\vec{z}}$ defined in Eq.~\eqref{eq:Gin}, which is induced from $G$ by the index vectors of $P_{\vec{x},\vec{z}}$.
\end{proposition}
The proof is left in Appendix \ref{App:PLcomp}.
By applying the above result of Eq.~\eqref{eq:PLcomp} to the definition of Eq.~\eqref{eq:malpha} and \eqref{eq:SREalpha} and some further simplification, one has the following general formula for the magic of quantum hypergraph states.
\begin{theorem}\label{th:PLmom}
The $\alpha$-order PL-moment of a hypergraph state $\ket{G}$ shows
\begin{equation}\label{eq:PLmoment}
\begin{aligned}
    \mathbf{m}_{\alpha}(\ket{G})=2^{-n(1+2\alpha)}\sum_{\vec{x},\vec{z}}\Tr{U(G^{*}_{\vec{x},\vec{z}})}^{2\alpha},
\end{aligned}
\end{equation} 
where $U(G^{*}_{\vec{x},\vec{z}})$ is the phase unitary determined by the hypergraph $G^{*}_{\vec{x},\vec{z}}$ defined in Eq.~\eqref{eq:GinSim}. The corresponding 
\renyiaa\ reads 
\begin{equation}\label{eq:SREfrist}
   \mathbf{M}_{\alpha}(\ket{G})=\frac{1+2\alpha}{\alpha-1}n+\frac1{1-\alpha}\log \sum_{\vec{x},\vec{z}}\Tr{U(G^{*}_{\vec{x},\vec{z}})}^{2\alpha}.
\end{equation}	
\end{theorem}

Compared to Eq.~\eqref{eq:PLcomp} with the phase unitary $U(G_{\vec{x},\vec{z}})$, Eq.~\eqref{eq:PLmoment} and \eqref{eq:SREfrist} are only related to simplified one $U(G^{*}_{\vec{x},\vec{z}})$. This is based on the fact that the summation over all possible $\vec{z}$ simplifies the summation over all possible $E^{(1)}_{\vec{x},\vec{z}}$ to the one over $E^{(1)}_{\vec{z}}$. We remark that Ref.~\cite{lami2023quantum} also studied the magic of quantum hypergraph states mainly using the measure named min-relative entropy $\mathbf{D}_{\mathrm{min}}(\ket{G})$, which is related to the maximal fidelity to stabilizer states. In particular, they relate an upper bound of $\mathbf{D}_{\mathrm{min}}(\ket{G})$ to a minimization procedure of the nonquadraticity of the resulting Boolean function. However, the minimization result can only be obtained for a few example states. On the other hand, $\mathbf{M}_{\alpha}(\ket{G})$ shown here can be used as an lower bound of $\mathbf{D}_{\mathrm{min}}(\ket{G})$ \cite{Leone2022SRE,haug2023stabilizer}.

Theorem \ref{th:PLmom} transforms the PL-moment and also SRE into the calculation of trace of a family of phase unitary. There are totally $4^n$ such kind of induced $U(G^{*}_{\vec{x},\vec{z}})$, which makes the calculation of the magic for a given hypergraph state still challenging.

Specifically, for $3$-uniform hypergraph states $\ket{G_{n,3}}$, the hyperedges in $E^{(2)}_{\vec{x}}$ are all $2$-edges so the corresponding gates are $CZ$ gates. Then PL-moment in Eq.~\eqref{eq:PLmoment} becomes a summation of Boolean functions, for $\alpha=2$, it reads explicitly as
\begin{equation}\label{eq:3magic}
    \mathbf{m}_{2}(\ket{G_{n,3}})=2^{-5n}\sum_{\vec{x},\vec{z}}\left(\sum_{\vec{a}}(-1)^{\sum_i z_i a_i+\sum_{j,k}b_{j,k}(\vec{x}) a_j a_k}\right)^4
\end{equation}
where $b_{j,k}(\vec{x})=\sum_{\{v_i,v_j,v_k\}\in E}x_i$ represents the summation of all the possible $CZ$ gates for a fix qubit-pair $\{j,k\}$.
Naively, there are totally $8^n$ terms to sum the indices of $\vec{x},\vec{z}$ and also $\vec{a}$, which looks very sophisticated. 

Nevertheless, based on the pictural expression in Theorem \ref{th:PLmom}, we can give an upper bound of SRE of general hypergraph states with respect to its average degree. 
\begin{theorem}\label{th:bdegree}
For any $n$-qubit hypergraph state $\ket{G}$ whose corresponding graph $G$ has average degree $\bar{\Delta}(G)$, its S$R_\alpha$E with $\alpha \geq 2$ is upper bounded by
\begin{equation}\label{eqth:bdegree}
    \mathbf{M}_{\alpha}(\ket{G})\leq \frac{1}{\alpha-1}\left[1-\log\left(1+\frac{1}{2^{(2\alpha-1)\bar{\Delta}(G)}}\right)\right]n.
\end{equation}
\end{theorem}
The proof is left in Appendix \ref{App:Thdegree}. For $\alpha=2$, and a large enough $\bar{\Delta}(G)$, the upper bound in Eq.~\eqref{eqth:bdegree} behaves like $[1-\log(e)^{-1}2^{-3\bar{\Delta}(G)}]n$. This indicates that for a quantum hypergraph state with $\bar{\Delta}(G)$ bounded by some constant, its magic cannot reach the maximum possible value $n-o(n)$. This bound especially constrains the magic of hypergraph states on the lattice. Moreover, a more direct upper bound for any hypergraph state shows $\mathbf{M}_{\alpha}(\ket{G}) \leq \frac{n}{\alpha-1}$. It is interesting to remark that $\mathbf{M}_{\alpha}(\ket{G})$ of $\alpha>2$ can not reach the value near $n$.

\section{Magic of random hypergraph states}\label{sec:random}
In this section, we study the statistical properties magic of random hypergraph states. First, 
we define some random hypergraph state ensembles $\mc{E}$ from the corresponding random hypergraph ensembles. Here, we mainly study random $c$-uniform hypergraphs, which only own $c$-edge. A random $c$-uniform hypergraph ensemble can be determined by the probability $p$ whether there is a $c$-edge or not among all choices of $c$ vertices. 
Denote the combination number $C^c_n:=\binom{n}{c}$, and the ensembles are defined formally as follows \cite{zhou2022hyper}.

\begin{definition}\label{Def:ensemble}
The ($c$-uniform) random hypergraph state ensemble $\mathcal{E}_c^p$ of $n$-qubit system is defined as
\begin{equation}\label{Eq:Def:ensemble}
\begin{aligned}
\mathcal{E}_c^p=\left\{\ket{\Psi}=U\ket{\Psi_0}\Big|U=U_{e_{C_{n}^c}}\cdots U_{e_2}U_{e_1}\right\}.
\end{aligned}
\end{equation}
Here each $e_i$ is a distinct $c$-edge of the $n$ vertices, with totally $C^c_n$ such edges,  $U_{e_i}$ acts on the Hilbert space $\mc{H}_{e_i}$ by taking $\{\id_{e_i},CZ_{e_i}\}$ from the probability distribution $\{1-p,p\}$ respectively, and the initial state $\ket{\Psi_0}=\ket{+}^{\otimes N}$.
\end{definition}
Note that the specific gate sequence in Eq.~\eqref{Eq:Def:ensemble} is not relevant, as $CZ_{e}$ gates commute with each other. In particular, as $p=1/2$, all the elements in $\mc{E}_k$ 
share an equal probability $1/|\mc{E}_c|$ with $|\mc{E}_c|=2^{C_{n}^c}$ \cite{zhou2022hyper}. Hereafter we omit the superscript $p$ of $\mc{E}_c^p$ as $p=1/2$ for simplicity of notation. See Fig.~\ref{fig:randomCons} (a) for an example of a $3$-uniform hypergraph state. 

The SRE defined in Eq.~\eqref{eq:SREalpha} is in the logarithmic function. As such, to analyze the average property of SRE for some state ensembles $\mc{E}$,  we instead focus on the calculation of the average PL-moment of Eq.~\eqref{eq:malpha} inside the logarithm, that is,
\begin{equation}\label{eq:malphaAv}
    \langle\mathbf{m}_{\alpha}\rangle_{\mc{E}}=\mbb{E}_{\Psi\in \mc{E}}[\mathbf{m}_{\alpha}(\ket{\Psi})].
\end{equation}
This directly gives a lower bound of the \renyiaa\ by the concavity of the logarithmic function. 
\begin{equation}\label{eq:SREalphaAV}
\langle\mathbf{M}_{\alpha}\rangle_{{\mc{E}}}\geq (1-\alpha)^{-1}\log \langle\mathbf{m}_{\alpha}\rangle_{\mc{E}}.
\end{equation}	
for $\alpha>1$.

In particular, for $p=1/2$, one has the average PL-moment of the hypergraph ensemble $\mc{E}_c$ defined in Eq.~\eqref{Eq:Def:ensemble} as
\begin{equation}
\langle\mathbf{m}_{\alpha}\rangle_{\mc{E}_c}=|\mc{E}_c|^{-1}\sum_{G_{n,c}\in\mathcal{G}_{n,c}}\mathbf{m}_\alpha(\ket{G_{n,c}}),
\end{equation}
where $\mathcal{G}_{n,c}=\{G\big||V|=n, \forall e\in E, |e|=c\}$ is the set of $c$-uniform hypergraphs. 

Hereafter, we focus on $\alpha \geq 2 \in \mathbb{Z}$ and the uniform ensemble $\mc{E}_c$ to calculate the average properties and the fluctuations of the magic, and then extend to non-uniform ensemble $\mc{E}_c^p$ later. The adjustment of the probability parameter $p$ can thus change the expected density of the applied gates and also the expected average degree of the corresponding hypergraph, which is discussed in Sec.~\ref{sec:ppp}. 
We remark that our work is the first to show statistical properties of magic for some realistic ensembles beyond Haar random states \cite{Leone2022SRE,Liu2022manybodymagic,white2020mana}.

\subsection{Average analysis of magic}\label{sec:AV}
In this section, our main focus is on the average properties of magic, especially the PL-moment, and then we show quite tight lower bounds of the average SRE. The following theorem transforms the average PL-moment into a counting problem of binary strings. We first show some related definitions of the norm and operations of an $n$-bit string $\vec{t}=\{t^i\}$. The $1$-norm $\norm{\vec{t}}_1=\sum_i t^i$, with the addition modulo $2$. The Hadamard or Schur product $\bigodot$ of some bit strings $\vec{t_k}$ is the element-wise product, i.e., $\vec{t'}=\bigodot_{k}\vec{t_k}$ with ${t'}^i=\prod_k t_k^i$. 
\begin{theorem}\label{theo:randomave}
For any integer $\alpha \geq 2 \in \mathbb{Z}$, the average $\alpha$-th PL-moment of $n$-qubit random hypergraph state ensembles $\mc{E}_c$ defined in Eq.~\eqref{Eq:Def:ensemble} shows
\begin{equation}
    \langle\mathbf{m}_{\alpha}\rangle_{\mc{E}_c}=\frac{\mathrm{N}(c,\alpha,n)}{2^{2\alpha n}}
\end{equation}
where $\mathrm{N}(c,\alpha,n)$ is the number of $2$-tuple 
$(\mathrm{T},\vec{x})$, 
such that the following two constraints are satisfied.
\begin{align}
&\norm{\vec{t_i}}_1=0,\quad\forall i, \label{eq:conhyper1}\\
&\sum_{q\subset e_c}\left(\prod_{v_i\in q}x_i\norm{\bigodot_{v_k\in e_c\setminus q}\vec{t_k}}_1\right)=0,\quad\forall |e_c|=c, \label{eq:conhyper2}
\end{align}
where $\mathrm{T}=(\vec{t_1},\cdots,\vec{t_i},\cdots,\vec{t_n})$ is a $2\alpha\times n$ binary matrix, $\vec{x}$ is an $n$-bit vector with elements $x_i$, and $e_c$ labels all possible $c$-edges.
\end{theorem}

The full proof is left in Appendix \ref{App:Thmain-counting}, and here we give some intuition about the proof. We mainly utilize the replica trick to write the 
average PL-moment on $2\alpha$-copy of the original Hilbert space $\mc{H}_d^{\otimes 2\alpha}$ \cite{Leone2022SRE,Haug2022MPSmagic,zhou2022hyper}. Consequently, the computational basis of every qubit can be labeled by a $2\alpha$-bit string \cite{zhou2022hyper}, i.e., $\vec{t_i}$ of the matrix $\mathrm{T}$ for the $i$-th qubit. The two constraints Eq.~\eqref{eq:conhyper1} and Eq.~\eqref{eq:conhyper2}, which look a bit complicated at first glance, actually correspond to the induced hypergraph structure of Eq.~\eqref{eq:GinSim} previously introduced in Sec.~\ref{sec:SREG}.
In particular, the first constraint Eq.~\eqref{eq:conhyper1} accounts for the effect of the expectation from $2\alpha$-replica on the edge set $E^{(1)}_{\vec{z}}$, and the second one Eq.~\eqref{eq:conhyper2} for the edge set $E^{(2)}_{\vec{x}}$, that is, edges introducing multi-qubit controlled gates. See Fig.~\ref{fig:randomCons} (b) for an illustration of the constraints when $c=3$ and $\alpha=2$.

\begin{figure}[tbhp!]
\centering
\resizebox{10cm}{!}{\includegraphics[scale=0.8]{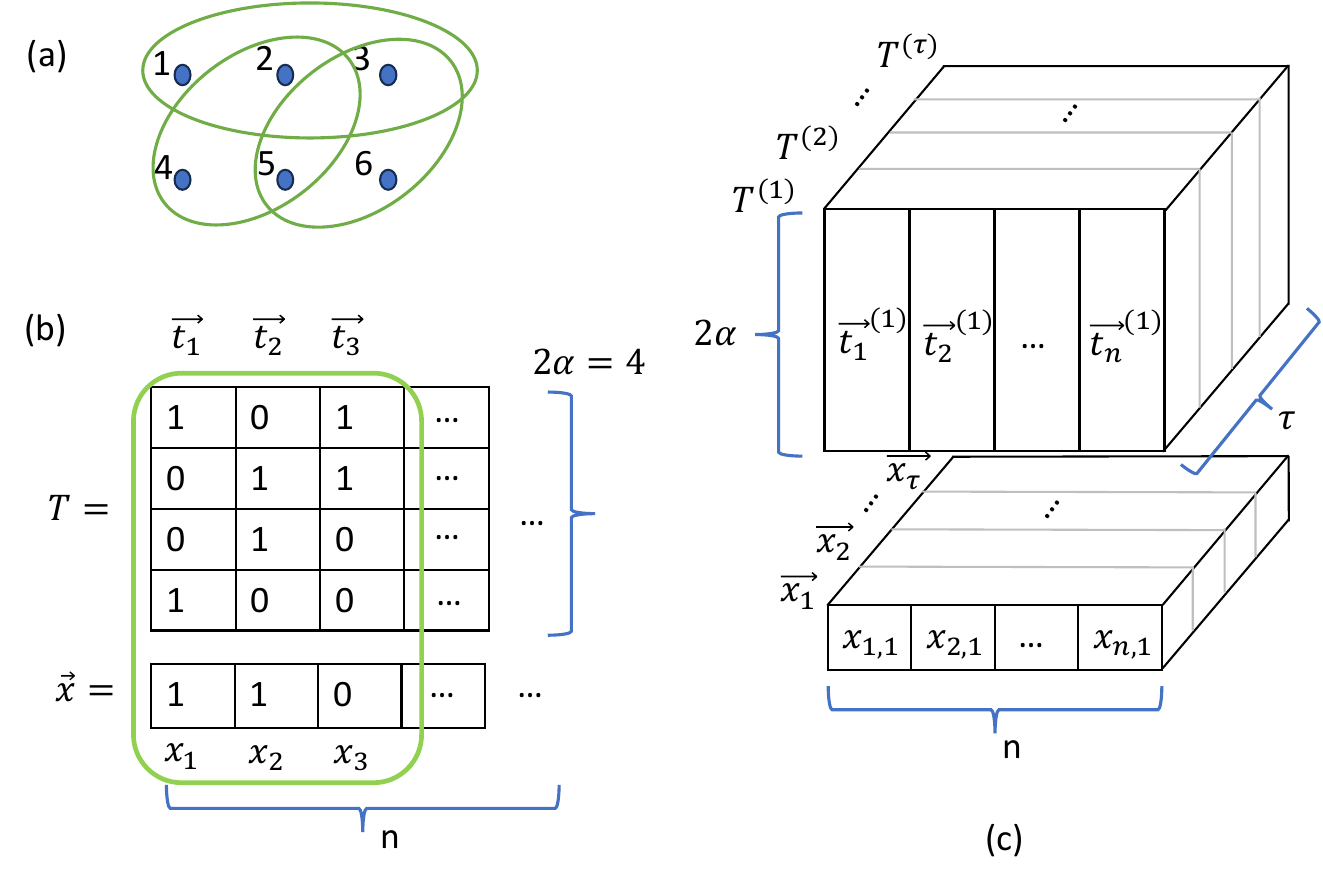}}
\caption{(a) A random $c=3$-unifrom hypergraph with 6 vertices. In this case, there are three $3$-edges selected, $\{1,2,3\},\{2,4,5\},\{3,5,6\}$.  The ensemble $\mathcal{E}_c^p$ in Def.~\ref{Def:ensemble} contains random hypergraph states, whose hypergraphs are generated by selecting all possible $c$-edges independently with probability $p$. (b) An illustration of the constraints in Thereom \ref{theo:randomave} and Corollary \ref{co:3special} with $c=3$ and $\alpha=2$. Here we explicitly give the first three 4-bit stings $\vec{t_1},\vec{t_2},\vec{t_3}$ with respective to the first three qubits, and clearly they satisfy the parity constraint in Eq.~\eqref{eq:conhyper1}. The second constraint in Eq.~\eqref{eq:conhyper2} is simplified to Eq.~\eqref{eq:conhyper3edge} as $c=3$. For a specific $3$-edge $e_3=\{1,2,3\}$, its induced constraint reads $x_1\norm{\vec{t_2}\bigodot\vec{t_3}}_1+x_2\norm{\vec{t_3}\bigodot\vec{t_1}}_1+x_3\norm{\vec{t_1}\bigodot\vec{t_2}}_1=0$, and it is satisfied if $\vec{x}=\{1,1,0,\cdots\}$. Note that the assignment of $\mathrm{T}=\{\vec{t_i}\}_{i=1}^n$ and $\vec{x}$ should satisfy all such constraint from every $3$-edge. (c) A graphic illustration of the 2-tuple $(\mathcal{T},\mathrm{X})$ in Thereom \ref{theo:randomvar}. $\mathcal{T}$ is a $2\alpha\times n\times \tau$ $3$-order tensor, and $\mathrm{X}$ is a $n\times \tau$ matrix. The first constraint in Eq.~\eqref{eq:convar1} is just the parity one for each $\vec{t_i}^{(s)}$ of $\mathcal{T}$. The second one in Eq.~\eqref{eq:convar2} is a generalization of Eq.~\eqref{eq:conhyper2}, by additionally summing the index $s\in [\tau]$. For every edge $e_c$, one first calculates Eq.~\eqref{eq:conhyper2} on each plane of this cube and then sums all of them in the longitudinal direction.}\label{fig:randomCons}
\end{figure}

Before showing refined results later, we first give some preliminary estimation of the counting problem in Theorem \ref{theo:randomave}. Note that Eq.~\eqref{eq:conhyper1} is indeed the parity constraint and independent of $\vec{x}$. 
Denote the set $\mc{D}=\left\{\vec{t}~\big|~\|\vec{t}\|_1=0\right\}$ of $2\alpha$-bit strings as the valid set, one directly has $|\mc{D}|=2^{2\alpha-1}$. Each $\vec{t_i}$ of $\mr{T}$ should be taken from $\mc{D}$, and consequently the entire matrix $\mr{T}$ comprises $\left(2^{2\alpha-1}\right)^n$ alternatives when Eq.~\eqref{eq:conhyper1} satisfied. Furthermore, as for Eq.~\eqref{eq:conhyper2}, one needs to find the number of $\vec{x}$ given a specific assignment of $\mathrm{T}$ with each $\vec{t_i} \in \mc{D}$.   Supposed that $\vec{x}=\vec{0}$, it is not hard to check that in this case, Eq.~\eqref{eq:conhyper2} induces no constraint on $\mr{T}$. If we ignore the constraint of Eq.~\eqref{eq:conhyper2}, it is clear that there are totally $2^n$ distinct $\vec{x}$. 
On account of these observations, we have a lower bound and a trivial upper bound of PL-moment as
\begin{equation}\label{eq:bounds}
    \frac{(2^{2\alpha-1})^n\cdot 1}{2^{2\alpha n}}=2^{-n}\leq\langle\mathbf{m}_{\alpha}\rangle_{\mc{E}_c}\leq \frac{(2^{2\alpha-1})^n\cdot 2^n}{2^{2\alpha n}}=1.
\end{equation}
The upper bound is just $1$ and is trivial by definition of moments. With a refined analysis of Eq.~\eqref{eq:conhyper2} and more exact counting of $\mathrm{N}(c,\alpha,n)$, a non-trivial upper bound of the PL-moment is shown as follows for general $\alpha$ and $c$.
\begin{proposition}\label{prop:general}
For any $c\geq 3 \in \mathbb{Z}$ and $\alpha\geq 2 \in \mathbb{Z}$, with the qubit number $n\gg\alpha$ and $n\gg c$, the average PL-moment
\begin{equation}\label{eq:genbound}
    \langle\mathbf{m}_{\alpha}\rangle_{\mc{E}_c}\leq \frac{2^{(c+2^{2\alpha-1})}}{2^n}.
\end{equation}
In particular, for $c=3$ and $\alpha=2$, it shows $\langle\mathbf{m}_{2}\rangle_{\mc{E}_3}\leq 2^{-(n-11)}$.
\end{proposition}

The upper bound in Proposition \ref{prop:general} is a bit loose considering the dependence on the parameters $c$ and $\alpha$, which is exponential to $c$ and even double-exponential to $\alpha$. However, suppose one only considers constant $c$ and $\alpha$ and focuses on the relation to the qubit number $n$, the upper bound then shows $\langle\mathbf{m}_{\alpha}\rangle_{\mc{E}_c}=O\left(2^{-n}\right)$, which matches the lower bound in Eq.~\eqref{eq:bounds}. We summarize this and its implication to SRE via Eq.~\eqref{eq:SREalphaAV} in the following corollary.

\begin{corollary}
    For any $c\geq 3 \in \mathbb{Z}$ and $\alpha\geq 2 \in \mathbb{Z}$, with $n\gg\alpha$ and $n\gg c$, the average SRE is lower bounded by
\begin{equation}\label{eq:genboundSRE}
    \langle\mathbf{M}_{\alpha}\rangle_{\mc{E}_c}\geq (\alpha-1)^{-1}[n-(c+2^{2\alpha-1})].
\end{equation}
In particular, for constant $c$ and $\alpha$, it shows that  $\langle\mathbf{M}_{\alpha}\rangle_{\mc{E}_c}\geq (\alpha-1)^{-1} n-O(1)$.
\end{corollary}

Note that as $\alpha=2$, the lower bound of average SRE shows that the magic of random hypergraph states is nearly the maximal possible value $n$ for any constant $c\geq 3$ and sufficiently large qubit number $n$, which reproduces the result of Haar random states \cite{Leone2022SRE}.

Furthermore, a more refined analysis can be conducted when focusing on $3$-uniform hypergraph states. A direct simplification of Eq.~\eqref{eq:conhyper2} in this case is shown as follows.
\begin{corollary}\label{co:3special}
    As $c=3$, the second constraint in Eq.~\eqref{eq:conhyper2} of Theorem \ref{theo:randomave} can be simplified to
    \begin{equation}\label{eq:conhyper3edge}
    \begin{gathered}
        \sum_{v_i\in e_3}\left(x_i\norm{\bigodot_{k\neq i, v_k\in e_3}\vec{t_k}}_1\right)=0,\quad\forall |e_3|=3.
    \end{gathered}
    \end{equation}
\end{corollary}
The proof is straightforward. As $c=3$, one only needs to sum over all possible $q$ with $|q|=1$ and $|q|=2$ in Eq.~\eqref{eq:conhyper2}. For $|q|=2$, the set $e_3\setminus q$ only has one element, so the $l_1$-norm is just for a single binary vector, which is already constrained by Eq.~\eqref{eq:conhyper1}. As a result, the constraint is reduced to that for $|q|=1$ as shown in Eq.~\eqref{eq:conhyper3edge}. Based on this simplification, we give the following more accurate counting results.
\begin{proposition}\label{prop:3cardi}
    The average PL-moments of random $3$-uniform hypergraph states are
    \begin{equation}
    \begin{aligned}
        &\langle\mathbf{m}_{2}\rangle_{\mc{E}_3}=\frac{7}{2^n}-\frac{14}{4^n}+\frac{8}{8^n},\label{eq:3cardiA2}\\
        &\langle\mathbf{m}_{\alpha}\rangle_{\mc{E}_3} \leq\frac{4}{2^n}+\frac{(2\alpha-1)!!}{2^{(\alpha-1)n}}, \quad\forall \alpha \geq 3 \in \mathbb{Z}.
    \end{aligned}
    \end{equation}
\end{proposition}

The proofs of Propositions \ref{prop:general} and \ref{prop:3cardi} are left in Appendix \ref{ApppropRecursion} and \ref{Appprop2}, respectively.



\subsection{Variance analysis of magic}\label{sec:VAR}
In this section, we further investigate the variance of PL-moments, which is helpful in constructing a sharp concentration result of SRE for random hypergraph states in the large $n$ limit. In this way, we find that hypergraph states which are easier to prepare, \emph{typically} own the maximal magic, similar to the Haar random states.

The variance of PL-moment is defined as

\begin{equation}\label{eq:malphaVar}
\begin{aligned}
\delta^2_\mc{E}[\mathbf{m}_{\alpha}]&:=\langle\mathbf{m}_{\alpha}^2\rangle_{\mc{E}}-\langle\mathbf{m}_{\alpha}\rangle_{\mc{E}}^2
    \end{aligned}
\end{equation}
with the average PL-moment $\langle\mathbf{m}_{\alpha}\rangle_{\mc{E}}$ given in Eq.~\eqref{eq:malphaAv}. The average PL-moment is investigated in the previous Sec.~\ref{sec:AV}, and hereafter, we focus on the analysis of the first term, i.e., the average of the $2$-th moment of the PL-moment, and then combine them to show the final variance.

As an analog and extension of Theorem \ref{theo:randomave}, 
for the general $\tau$-th moment of the PL-moment, we transform the calculation of its average on random hypergraph state ensembles into the following counting problem. 
\begin{theorem}\label{theo:randomvar}
For any integer $\alpha \geq 2 \in \mathbb{Z}$, the $\tau$-th moment of the $\alpha$-th PL-moment of $n$ qubit random hypergraph state ensembles $\mc{E}_c$ shows
\begin{equation}
\langle\mathbf{m}_{\alpha}^{\tau}\rangle_{\mc{E}_c}=\frac{\mathrm{N}^{(\tau)}(c,\alpha,n)}{2^{2\tau\alpha n}}
\end{equation}
where $\mathrm{N}^{(\tau)}(c,\alpha,n)$ is the number of $2$-tuple $(\mathcal{T},\mathrm{X} )$, such that the following two constraints are satisfied.
\begin{align}
    &\norm{\vec{t_i}^{(s)}}_1=0,\quad\forall i,s \label{eq:convar1} \\
    &\sum_{s=1}^{\tau} \sum_{q\subset e_c} \left(\prod_{v_{i}\in q}x_{i,s}\norm{\bigodot_{v_k\in e_c\setminus q}\vec{t_k}^{(s)}}_1\right)=0,\quad\forall |e_c|=c, \label{eq:convar2} 
\end{align}
where $\mathcal{T}$ is a $2\alpha\times n\times \tau$ rank-$3$ binary tensor,  with its element $2\alpha\times n$ matrix denoted by $\mathrm{T}^{(s)}=({\vec{t_1}}^{(s)},\cdots,\vec{t_i}^{(s)},\cdots,\vec{t_n}^{(s)})$ for $s\in[\tau]$, and the element  of the $i$-th binary vector $\vec{t_i}^{(s)}$ is $t^{(s)}_{j,i}$; $\mathrm{X}=(\vec{x_1},\cdots,\vec{x_s},\cdots,\vec{x_{\tau}})$ is an $n\times \tau$ binary matrix with elements $x_{i,s}$.
\end{theorem}
Note that Theorem \ref{theo:randomvar} reduces to Theorem \ref{theo:randomave} as $\tau=1$, and its proof is similar but more complicated, which is left in Appendix \ref{Apptheo:randomvar}. The counting problem here is more challenging even for $\tau=2$ of interest, compared to the one in Theorem \ref{theo:randomave} of Sec.~\ref{sec:AV}. Instead of showing a general result similar to Proposition \ref{prop:general}, we focus on the case where $\alpha=2$ and $c=3$ here and further give an upper bound for the variance of the PL-moment.
\begin{proposition}\label{prop:var}
    The variance of $2$-order PL-moment on the random $3$-uniform hypergraph state ensemble $\mc{E}_3$ is
    \begin{equation}
        \delta^2_{\mc{E}_3}[\mathbf{m}_{2}]=\langle\mathbf{m}_{2}^2\rangle_{\mc{E}_3}-\langle\mathbf{m}_{2}\rangle_{\mc{E}_3}^2\leq\frac{60}{2^{3n}}.
    \end{equation}
\end{proposition}
The proof is by estimating the counting result in Theorem \ref{theo:randomvar} for $\tau=2, \alpha=2$ and $c=3$, and then combining with $\langle\mathbf{m}_{2}\rangle_{\mc{E}_3}$ 
in Eq.~\eqref{eq:3cardiA2}, and we leave it in Appendix \ref{Appprop:var}. 
Note that the standard variance here $\delta_{\mc{E}_3}[\mathbf{m}_{2}]=O(2^{-1.5n})$ is exponentially small than its mean value $\langle\mathbf{m}_{\alpha}\rangle_{\mc{E}_3}=O(2^{-n})$ as shown in Eq.~\eqref{eq:3cardiA2}. Consequently, by utilizing Chebyshev's inequality, 
there is the concentration of measure effect of magic for the $\mc{E}_3$ ensemble. 
\begin{corollary}\label{co:concentarte}
    For a hypergraph state $\ket{G_{n,3}}$ of $n$-qubit chosen randomly from the $3$-unifrom hypergraph state ensemble $\mc{E}_3$, the probability that its \renyit\ larger than $n-3$ is almost $1$, that is
    \begin{equation}
   \Pr\left\{\mb{M}_2(\ket{G_{n,3}})\geq n-3\right\}\geq 1-\frac{60}{2^n}.
    \end{equation}
    \end{corollary}
A similar concentration effect of magic is also observed for Haar random states \cite{Leone2022SRE,Liu2022manybodymagic,white2020mana}. However, compared with Haar random states, $\mc{E}_3$ are quite easy to prepare, i.e., just by randomly operating three-qubit $CCZ$ gates.

\subsection{Average magic with general probability $p$}\label{sec:ppp}
In the previous two subsections, Sec.~\ref{sec:AV} and Sec.~\ref{sec:VAR}, we mainly focus on the magic properties of the hypergraph state ensemble $\mc{E}_c$. In this section, we extend the ensemble $\mc{E}_c$ to $\mc{E}_c^p$ in Definition \ref{Def:ensemble}, where each $c$-edge is randomly selected by the probability parameter $p$. In particular, we focus on the $c=3$ and $\alpha=2$ cases and show the following analytical formula for the average PL-moment for any $n$-qubit system.

\begin{theorem}\label{theo:p-ave}
    The average $2$-order PL-moment of the state ensemble $\mc{E}_3^p$, i.e., random $3$-uniform hypergraph state with each hyperedge selected by probability $p$, shows
  \begin{equation}\label{eq:pave}
   \langle\mathbf{m}_{2}\rangle_{\mc{E}_3^p}=2^{-3n}\sum_{\vec{\kappa}}{n \choose \vec{\kappa}}(1-2p)^{f(\vec{\kappa})}.
\end{equation}
Here $\vec{\kappa}=\left(\vec{\kappa^+},\vec{\kappa^-}\right)$ is a $8$-dimension vector composed of two $4$-dimension vectors $\vec{\kappa^+}$ and $\vec{\kappa^-}$, with each element $\kappa^{\pm}_{1\leq i \leq 4}\in \mbb{N}$ being non-negative integer; the summation of the elements of $\vec{\kappa}$ equals $\sum_{i=1}^4(\kappa_i^+ + \kappa_i^-)=n$, and the multinomial coefficient ${n \choose \vec{\kappa}}$ is short for the combinatorial number ${n \choose \kappa_0^+,\kappa_0^-,\kappa_1^+,\kappa_1^-,\kappa_2^+,\kappa_2^-,\kappa_3^+,\kappa_3^-}$; the function reads
\begin{equation}\label{eq:fa}
\begin{aligned}
    f(\vec{\kappa})=&\kappa_0^+(\kappa_1\kappa_2+\kappa_2\kappa_3+\kappa_3\kappa_1)\\
    &+\kappa_1^+\kappa_1^-(\kappa_2+\kappa_3)+\kappa_2^+\kappa_2^-(\kappa_3+\kappa_1)+\kappa_3^+\kappa_3^-(\kappa_1+\kappa_2)\\
    &+\kappa_1^+\kappa_2^-\kappa_3^-+\kappa_2^+\kappa_3^-\kappa_1^-+\kappa_3^+\kappa_1^-\kappa_2^-+\kappa_1^+\kappa_2^+\kappa_3^+
\end{aligned}
\end{equation}
with $\kappa_i=\kappa_i^++\kappa_i^-$ for short.
\end{theorem}
The proof is left in Appendix \ref{Apptheo:pave}. We remark that Theorem \ref{theo:p-ave} could be extended to arbitrary constant $c$ cases while the polynomial complexity of calculation holds since one still only needs to sum over some combinatorial number of the vector $\vec{\kappa}$. 

\begin{figure}[tbhp!]
\centering
\resizebox{9cm}{!}{\includegraphics[scale=0.8]{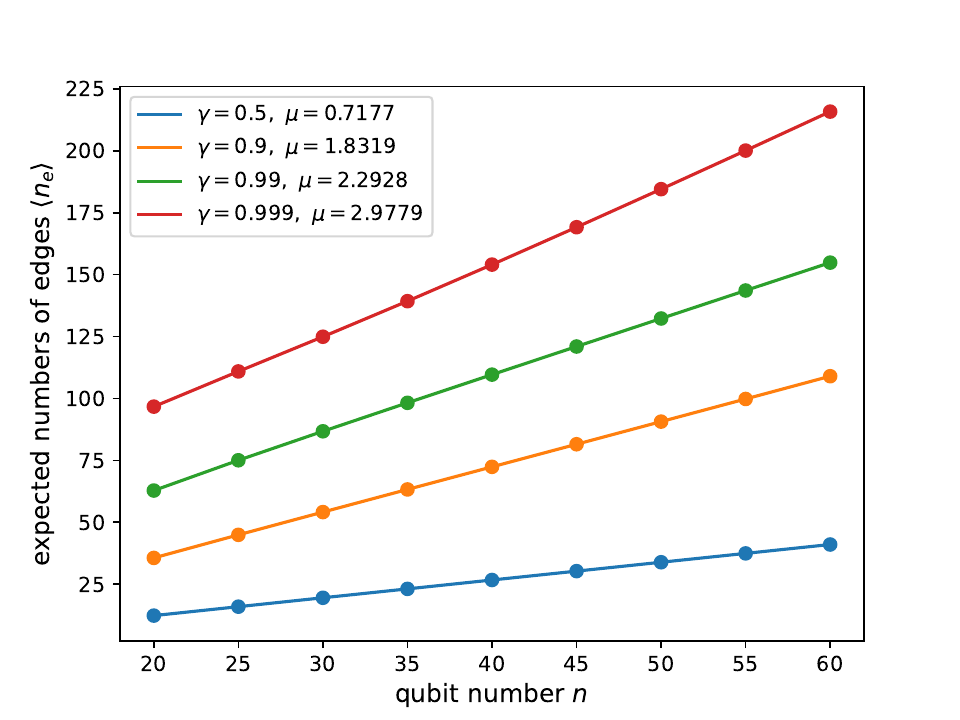}}
\caption{Relation between qubit number $n$ ($x$-axis) and expected number of hyperedges $\langle n_e\rangle=p\cdot{n \choose 3}$ ($y$-axis), given the (lower bound of) expected \renyit\ $\underline{\langle\mathbf{M}_{2}\rangle_{\mc{E}_3^p}}=-\log\langle\mathbf{m}_{2}\rangle_{\mc{E}_3^p}=\gamma\cdot n$ for different $\gamma$'s. Different colors represent different expected magic with different slopes $\mu$'s.}
  \label{fig:Pmagic} 
\end{figure}

Note that Eq.~\eqref{eq:3cardiA2} of the previous $p=1/2$ case is obtained by counting the number of solutions to $f(\vec{\kappa})=0$. For the general $p$ here, we do not have a very compact formula, and the function $f(\vec{\kappa})$ and also $\langle\mathbf{m}_{2}\rangle_{\mc{E}_3^p}$ in Eq.~\eqref{eq:pave} look a little tedious. Nevertheless, the calculation of $\langle\mathbf{m}_{2}\rangle_{\mc{E}_3^p}$ only involves polynomial complexity with respective to the qubit-number $n$, that is, $O(n^7\log n)$. This fact enables us to numerically study the relationship between the average magic and the probability $p$ for quite large $n$, as shown in Fig.~\ref{fig:Pmagic}. 

To be specific, Fig.~\ref{fig:Pmagic} shows the relation between the expected number of hyperedges $\langle n_e\rangle:=p\cdot{n \choose 3}$ and the qubit-number $n$, given the average magic $\langle\mathbf{m}_{2}\rangle_{\mc{E}_3^p}=\gamma n$ for some constant $\gamma$. One can see that for a fixed proportion $\gamma$ of $n$, $\langle n_e\rangle$ is almost linear to $n$ for different $\gamma$'s, i.e., $\langle n_e\rangle \sim \mu n$, and thus $p\sim O(n^{-2})$ far less than $1/2$ like before.

For each vertex, the expected number of edges is  
\begin{equation}\label{eq:}
\begin{aligned}
    \langle n_e\rangle \cdot \frac{{n-1\choose 2}}{{n \choose 3}}\sim 3\mu,
\end{aligned}
\end{equation}
and thus the expected average degree $\langle \bar{\Delta}(G) \rangle$ is about a constant. This shows the consistency to Theorem \ref{th:bdegree}, where the magic of a bounded-average-degree hypergraph state is also bounded. For $\gamma=0.999$, which is very near the maximal $1$, the slope $\mu \simeq 3.0$. It means that a very small $p$ can let the average magic near the maximal value.

The statistical results here may also suggest a dynamical way to generate maximal magic states efficiently. For each step, one operates a $CCZ$ gate on any three-qubit chosen randomly from a $3$-edge, and repeats this process for about $K=O(n)$ times. In particular, the numerical result implies that $K=3n$ may be enough to let \renyit\ reach $0.999n$. Moreover, if one parallel applies $CCZ$ gates, constant-depth (depth-$3\mu$) quantum circuit could be sufficient (depth-9 for $\mathbf{M}_{2}(\ket{G})\sim 0.999n$).

\section{Magic of some symmetric hypergraph states}
In this section, we investigate the magic of some hypergraph states with \emph{permutation symmetry}. We focus on the $c$-complete hypergraph states with $c=3$ and $c=n$, and our method can be applied to other symmetric hypergraph states.

The motivation for this study is two-fold. First, for any $c$-complete hypergraph, every vertex is connected to other $n-1$ vertices, and thus the average degree being $\bar{\Delta}(G_{c\textrm{-com}})=n-1$. Consequently, $c$-complete hypergraph states can give some new insight beyond the one with a bounded average degree whose magic is constrained by Theorem \ref{th:bdegree}.
Second, the permutation symmetry significantly simplifies the summation of indices ${\vec{x},\vec{z}}$ in Eq.~\eqref{eq:PLmoment} from exponential to polynomial. This fact and the pictural magic formula introduced in Theorem \ref{th:PLmom} enable calculating the spectrum of the PL-components, and thus \renyiaa\ for any $\alpha$. In particular, $\alpha$ needs not to be limited to integers as that in Sec.~\ref{sec:random}, for instance, here one can take $\alpha=1/2$.
\begin{definition}\label{Def:per}
An $n$-partite quantum state $\ket{\Psi}$ owns permutation symmetry, if 
\begin{equation}
    \begin{aligned}
        U_{\{\pi\}}\kb{\Psi}U^\dag_{\{\pi\}}=\kb{\Psi},\ \forall \pi\in S_n.
    \end{aligned}
    \end{equation}
Here $\pi$ is the permutation element in the $n$-th order symmetric group $S_n$. $U_{\{\pi\}}$ is an unitary representation with 
$U_{\{\pi\}}\ket{a_1,a_2,\cdots,a_n}=\ket{a_{\pi(1)},a_{\pi(2)},\cdots,a_{\pi(n)}}$
with $\{\ket{a}\}$ the basis state for a single-party.  For example, $U_{\{(1,2)\}}$ is the swap operator on 2-party.
\end{definition}
Due to the permutation symmetry, the PL-component $\Tr{P_{\vec{x},\vec{z}}\kb{\Psi}}$ is not relevant to the specific positions of single-qubit Pauli operators but only depends on the numbers of them, i.e., $x_i=1$ and $z_i=1$ and both $x_i=z_i=1$ in $P_{\vec{x},\vec{z}}$. Denote the corresponding sets of $P_{\vec{x},\vec{z}}$ as $A_x=\{i|x_i=1\}$, $A_{-x}=\{i|x_i=0\}$, $A_{z,x}=\{i|z_i=1,x_i=1\}$, $A_{z,-x}=\{i|z_i=1,x_i=0\}$, one has the following general observation. 

\begin{observation}\label{ob:sym}
For an $n$-partite quantum state $\ket{\Psi}$ owning permutation symmetry defined in Def.~\ref{Def:per}, its PL-component  $\Tr{P_{\vec{x},\vec{z}}\kb{\Psi}}$ is directly related to the cardinality of the sets $m=|A_x|$, $m_1=|A_{z,x}|$ and $m_0=|A_{z,-x}|$, but not the specific positions. That is, two components take the same value if these sets share the same cardinality.
\end{observation}

\begin{figure}[tbhp!]
\centering
\resizebox{9cm}{!}{\includegraphics[scale=0.8]{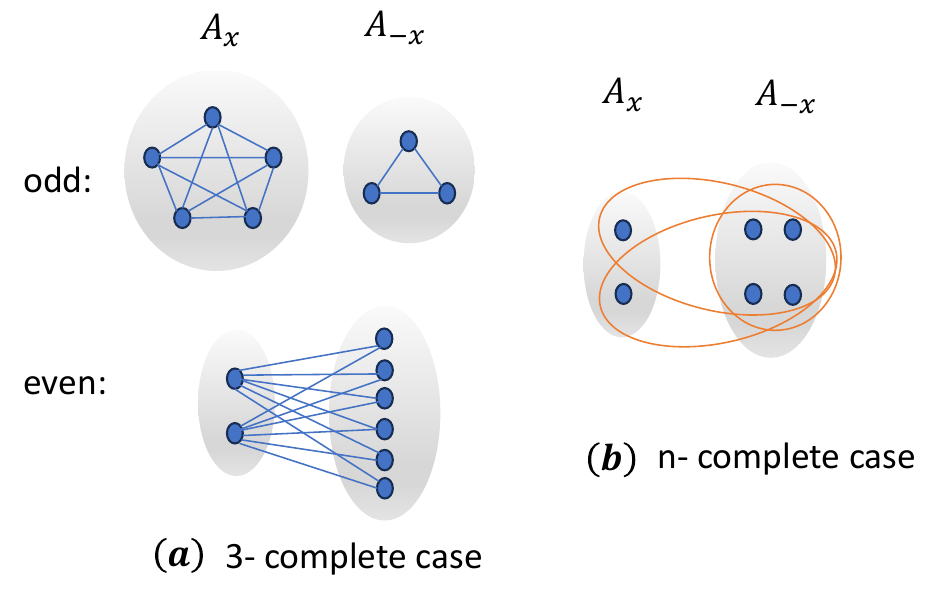}}
\caption{Illustration of induced hypergraphs for $c$-complete hypergraph states. Here the shaded areas label the corresponding sets. Due to the permutation symmetry, the induced hypergraph  $G^*_{\vec{x},\vec{z}}$ defined in Eq.~\eqref{eq:Gin} and \eqref{eq:GinSim} is determined by  $|A_x|$, $|A_{z,x}|$ and $|A_{z,-x}|$. Since $E^{(1)}_{\vec{z}}$ in Eq.~\eqref{eq:GinSim} only introduces 1-edges, for simplicity, we take $\vec{z}=0$ and thus $|A_{z,x}|=|A_{z,-x}|=0$. In this way, the structure of $G^*_{\vec{x},\vec{0}}$ only depends on $m=|A_x|$. (a) $3$-complete hypergraph with 8 vertices: $G^*_{\vec{x},\vec{0}}$ is related to even/odd of $m$, and it is either \emph{two complete-graphs} of $A_x$ and $A_{-x}$  or \emph{complete bipartite graph} between them.  (b) $n$-complete hypergraph with 6 vertices with $m=2$: $G^*_{\vec{x},\vec{0}}$ in this case contains all the hyperedges $A_{-x}\subseteq e'\subset V$. See Appendix \ref{app:sym} for more detailed discussions.
} \label{fig:indGsym} 
\end{figure}

For quantum hypergraph states whose PL-component is proportional to $\Tr{U(G_{\vec{x},\vec{z}})}$ as shown in Proposition \ref{prop:PLcomp}, if they own permutation symmetry, like $c$-complete hypergraph states, the induced hypergraphs $G_{\vec{x},\vec{z}}$ (and also $G^*_{\vec{x},\vec{z}}$) are isomorphic to each other if the aforementioned sets share the same cardinality. This observation and the pictural magic formula in Theorem \ref{th:PLmom} make the calculations intuitive and elegant. 

In the following, we first show the results of magic for $3$- and $n$-complete hypergraph states, respectively, and then give a series of comments on them. We leave the detailed proofs in Appendix \ref{app:sym}. We show their induced hypergraphs $G^*_{\vec{x},\vec{z}}$ in Fig.~\ref{fig:indGsym}.

\begin{proposition}\label{th:3-com}
The $2$-order and $1/2$-order PL-moment of the quantum $3$-complete hypergraph state of $n$-qubit $\ket{G_{3\textrm{-com}}}$ are
\begin{equation}\label{eq:PL3com}
\begin{aligned}
&\mathbf{m}_{2}(\ket{G_{3\textrm{-com}}})=\frac{1}{8}+\frac{7}{2^{n+\frac{3-(-1)^n}{2}}},\\
&\mathbf{m}_{1/2}(\ket{G_{3\textrm{-com}}})=2^{\frac{2n-7-(-1)^n}{4}}+1-2^{-n+\frac{1+(-1)^n}{2}}.
\end{aligned}
\end{equation}

Consequently, the corresponding $SRE$ approaches 
$\mathbf{M}_{2}(\ket{G_{3\textrm{-com}}})\rightarrow 3$ and 
   $\mathbf{M}_{1/2}(\ket{G_{3\textrm{-com}}})\rightarrow n-\frac{7+(-1)^n}{2}$, as $n\rightarrow\infty$. 
\end{proposition}

\begin{proposition}\label{th:n-com}
The $2$-order and $1/2$-order PL-moment of the quantum $n$-complete hypergraph state of $n$-qubit $\ket{G_{n\textrm{-com}}}$ are
\begin{equation}\label{eq:PLncom}
\begin{aligned}
    &\mathbf{m}_{2}(\ket{G_{n\textrm{-com}}})=1-16\cdot2^{-n}+112\cdot2^{-2n}
    -224\cdot2^{-3n}+128\cdot2^{-4n},\\
     &\mathbf{m}_{1/2}(\ket{G_{n\textrm{-com}}})=3-10\cdot2^{-n}+8\cdot2^{-2n},
\end{aligned}
\end{equation}
for $n\geq 2$.
Consequently, the corresponding SRE $\mathbf{M}_{2}(\ket{G_{n\textrm{-com}}})$ shows its maximal value $\log{\frac{32}{11}}\approx1.54$ when $n=3$, and decreases exponentially with $n$; $\mathbf{M}_{1/2}(\ket{G_{n\textrm{-com}}})$ increases monotonously to $2\log_2{3}$ as $n\rightarrow\infty$.
\end{proposition}

First, these two examples indicate that hypergraph states with an unbounded average degree could not own maximum magic, opposite to Theorem \ref{th:bdegree}. For $\ket{G_{3\textrm{-com}}}$, \renyit\ is about a constant, even though there are extensive $CCZ$ gates to generate it. This phenomenon also appears in entanglement analysis, where over-connected graphs and hypergraph states only own a small amount of entanglement. For instance, the $2$-complete graph state is equivalent to the GHZ state, with entanglement entropy being $1$ for any bipartition. 

For $\ket{G_{n\textrm{-com}}}$, its \renyit\ even scales like $O(2^{-n})$ for the large $n$ limit. This behavior can be understood via the fidelity to the initial stabilizer state $\ket{\Psi_0}$ with zero magic. The $n$-qubit gate $CZ_n$ is nearly an identity operator for a large $n$, and the fidelity $F(\ket{G_{n\textrm{-com}}},\ket{\Psi_0})=\bra{G_{n\textrm{-com}}}\Psi_0 \rangle=1-2^{-n}$. As such, $G_{n\textrm{-com}}$ should have small magic as being close to a stabilizer state. 

Second, the \renyit\ and \renyih\ are very different for both states. In particular, $\mathbf{M}_{1/2}(\ket{G_{3\textrm{-com}}})$ is near $n$, however, $\mathbf{M}_{2}(\ket{G_{3\textrm{-com}}})$ is about a constant. This is mainly due to the distribution of the PL-component of $\ket{G_{3\textrm{-com}}}$. There are almost half of PL-components are non-zero, but their distribution is very sharp. That is, most of them are extraordinarily small, and on the contrary, only a constant number of them are larger as $O(1)$. This phenomenon may lead to significant consequences in further one-shot manipulations \cite{tomamichel2012framework,Liu2019OneShot}.

\section{Conclusion and outlook}
In this work, we study the magic properties of quantum hypergraph states in terms of SRE. We find that the random hypergraph states typically have near-maximal magic, indicating an efficient way to generate them by dynamically applying few-qubit controlled phase gates. We also show a general upper bound of magic by the degree of the corresponding hypergraph and give a few analytical solutions to states with permutation symmetry. The general pictural formula of magic presented here can advance further explorations of many-qubit magic for large-scale systems.

From the results shown here, there are several directions to investigate further. First, it is direct to generalize the current study to other phase diagonal states \cite{Kruszynska2009Local,Nakata2014review,iaconis2021quantum}, especially those generated by 2-qubit $CZ_{\theta}$ gates. $CZ_{\theta}$ gate is more feasible in practise than $3$-qubit $CCZ$ gate and still can create magic. Second, the permutation symmetry significantly simplifies the calculation of magic, leading to a few analytical results, and it is promising to generalize to other states beyond hypergraph states, for example, the W-state \cite{odavic2022complexity}. It is thus interesting to analyze the statistical properties of state ensembles with high symmetry \cite{nakata2020generic}. Third, there is a large gap of \renyiaa\ for different $\alpha$’s, and this indicates the significance of analyzing the spectrum of the PL-component, which may play an important role similar to the entanglement spectrum \cite{Haldane2008Spectrum,Chamon2014Irreversibility,shaffer2014irreversibility}. And the implications of this phenomenon to one-shot quantum resource theory \cite{Liu2019OneShot} and magic distillation of many-qubit states \cite{bao2022magic} also deserve further studies. Finally, it is known that typically random hypergraph states also have near-maximal entanglement \cite{zhou2022hyper}. Thus the following important question arises---whether quantum states with near-maximal magic necessarily own a large amount of entanglement? The random states and permutation-symmetric states here give some positive support. Or, more broadly, what are the general relations and constraints between these two key quantum resources \cite{Leone2021Chaos,tirrito2023quantifying}? The answers to these questions would definitely unveil some ultra-quantum features and benefit quantum information processing.

\section{Acknowledgement}
We thank Alioscia Hamma, Xiongfeng Ma, and Ryuji Takagi for their useful discussions. J.~C and Y.~Y acknowledge the support of the National Natural Science Foundation of China (NSFC) via Grant No.12174216 and the Innovation Program for Quantum Science and Technology Grant via No.2021ZD0300804. Y.~Z acknowledges the support of National Natural Science Foundation of China (NSFC) Grant No.12205048, Innovation Program for Quantum Science and Technology 2021ZD0302000, the start-up funding of Fudan University, and the CPS-Huawei MindSpore Fellowship.



\begin{thebibliography}{79}
\providecommand{\natexlab}[1]{#1}
\providecommand{\url}[1]{\texttt{#1}}
\expandafter\ifx\csname urlstyle\endcsname\relax
  \providecommand{\doi}[1]{doi: #1}\else
  \providecommand{\doi}{doi: \begingroup \urlstyle{rm}\Url}\fi

\bibitem[Terhal(2015)]{Terhal2015QECC}
Barbara~M. Terhal.
\newblock Quantum error correction for quantum memories.
\newblock \emph{Rev. Mod. Phys.}, 87:\penalty0 307--346, Apr 2015.
\newblock \doi{10.1103/RevModPhys.87.307}.

\bibitem[Bravyi and Kitaev(2005)]{Bravyi2005Magic}
Sergey Bravyi and Alexei Kitaev.
\newblock Universal quantum computation with ideal clifford gates and noisy
  ancillas.
\newblock \emph{Phys. Rev. A}, 71:\penalty0 022316, Feb 2005.
\newblock \doi{10.1103/PhysRevA.71.022316}.

\bibitem[Veitch et~al.(2014)Veitch, Mousavian, Gottesman, and
  Emerson]{veitch2014resource}
Victor Veitch, SA~Hamed Mousavian, Daniel Gottesman, and Joseph Emerson.
\newblock The resource theory of stabilizer quantum computation.
\newblock \emph{New Journal of Physics}, 16\penalty0 (1):\penalty0 013009,
  2014.
\newblock \doi{10.1088/1367-2630/16/1/013009}.

\bibitem[Campbell et~al.(2017)Campbell, Terhal, and Vuillot]{campbell2017roads}
Earl~T Campbell, Barbara~M Terhal, and Christophe Vuillot.
\newblock Roads towards fault-tolerant universal quantum computation.
\newblock \emph{Nature}, 549\penalty0 (7671):\penalty0 172--179, 2017.
\newblock \doi{10.1038/nature23460}.

\bibitem[Gottesman and Chuang(1999)]{gottesman1999demonstrating}
Daniel Gottesman and Isaac~L Chuang.
\newblock Demonstrating the viability of universal quantum computation using
  teleportation and single-qubit operations.
\newblock \emph{Nature}, 402\penalty0 (6760):\penalty0 390--393, 1999.
\newblock \doi{10.1038/46503}.

\bibitem[Horodecki et~al.(2009)Horodecki, Horodecki, Horodecki, and
  Horodecki]{Horodecki2007Entanglement}
Ryszard Horodecki, Pawe\l{} Horodecki, Micha\l{} Horodecki, and Karol
  Horodecki.
\newblock Quantum entanglement.
\newblock \emph{Rev. Mod. Phys.}, 81:\penalty0 865--942, Jun 2009.
\newblock \doi{10.1103/RevModPhys.81.865}.

\bibitem[Knill(2005)]{knill2005quantum}
Emanuel Knill.
\newblock Quantum computing with realistically noisy devices.
\newblock \emph{Nature}, 434\penalty0 (7029):\penalty0 39--44, 2005.
\newblock \doi{10.1038/nature03350}.

\bibitem[Aaronson and Gottesman(2004)]{Aaronson2004ImprovedGKT}
Scott Aaronson and Daniel Gottesman.
\newblock Improved simulation of stabilizer circuits.
\newblock \emph{Phys. Rev. A}, 70:\penalty0 052328, Nov 2004.
\newblock \doi{10.1103/PhysRevA.70.052328}.

\bibitem[Bravyi et~al.(2016)Bravyi, Smith, and Smolin]{Bravyi2016Trading}
Sergey Bravyi, Graeme Smith, and John~A. Smolin.
\newblock Trading classical and quantum computational resources.
\newblock \emph{Phys. Rev. X}, 6:\penalty0 021043, Jun 2016.
\newblock \doi{10.1103/PhysRevX.6.021043}.

\bibitem[Bravyi and Gosset(2016)]{Bravyi2016Improved}
Sergey Bravyi and David Gosset.
\newblock Improved classical simulation of quantum circuits dominated by
  clifford gates.
\newblock \emph{Phys. Rev. Lett.}, 116:\penalty0 250501, Jun 2016.
\newblock \doi{10.1103/PhysRevLett.116.250501}.

\bibitem[Bravyi et~al.(2019)Bravyi, Browne, Calpin, Campbell, Gosset, and
  Howard]{Bravyi2019simulation}
Sergey Bravyi, Dan Browne, Padraic Calpin, Earl Campbell, David Gosset, and
  Mark Howard.
\newblock Simulation of quantum circuits by low-rank stabilizer decompositions.
\newblock \emph{{Quantum}}, 3:\penalty0 181, September 2019.
\newblock ISSN 2521-327X.
\newblock \doi{10.22331/q-2019-09-02-181}.

\bibitem[Pashayan et~al.(2015)Pashayan, Wallman, and Bartlett]{Pashayan2015neg}
Hakop Pashayan, Joel~J. Wallman, and Stephen~D. Bartlett.
\newblock Estimating outcome probabilities of quantum circuits using
  quasiprobabilities.
\newblock \emph{Phys. Rev. Lett.}, 115:\penalty0 070501, Aug 2015.
\newblock \doi{10.1103/PhysRevLett.115.070501}.

\bibitem[Howard and Campbell(2017)]{Howard2017Resource}
Mark Howard and Earl Campbell.
\newblock Application of a resource theory for magic states to fault-tolerant
  quantum computing.
\newblock \emph{Phys. Rev. Lett.}, 118:\penalty0 090501, Mar 2017.
\newblock \doi{10.1103/PhysRevLett.118.090501}.

\bibitem[Seddon et~al.(2021)Seddon, Regula, Pashayan, Ouyang, and
  Campbell]{Seddon2021Speedups}
James~R. Seddon, Bartosz Regula, Hakop Pashayan, Yingkai Ouyang, and Earl~T.
  Campbell.
\newblock Quantifying quantum speedups: Improved classical simulation from
  tighter magic monotones.
\newblock \emph{PRX Quantum}, 2:\penalty0 010345, Mar 2021.
\newblock \doi{10.1103/PRXQuantum.2.010345}.

\bibitem[Amico et~al.(2008)Amico, Fazio, Osterloh, and
  Vedral]{Amico2008Entanglement}
Luigi Amico, Rosario Fazio, Andreas Osterloh, and Vlatko Vedral.
\newblock Entanglement in many-body systems.
\newblock \emph{Rev. Mod. Phys.}, 80:\penalty0 517--576, May 2008.
\newblock \doi{10.1103/RevModPhys.80.517}.

\bibitem[Sarkar et~al.(2020)Sarkar, Mukhopadhyay, and Bayat]{Sarkar2020NJP}
S~Sarkar, C~Mukhopadhyay, and A~Bayat.
\newblock Characterization of an operational quantum resource in a critical
  many-body system.
\newblock \emph{New Journal of Physics}, 22\penalty0 (8):\penalty0 083077, aug
  2020.
\newblock \doi{10.1088/1367-2630/aba919}.

\bibitem[Ellison et~al.(2021)Ellison, Kato, Liu, and
  Hsieh]{ellison2021symmetry}
Tyler~D Ellison, Kohtaro Kato, Zi-Wen Liu, and Timothy~H Hsieh.
\newblock Symmetry-protected sign problem and magic in quantum phases of
  matter.
\newblock \emph{Quantum}, 5:\penalty0 612, 2021.
\newblock \doi{10.22331/q-2021-12-28-612}.

\bibitem[Liu and Winter(2022)]{Liu2022manybodymagic}
Zi-Wen Liu and Andreas Winter.
\newblock Many-body quantum magic.
\newblock \emph{PRX Quantum}, 3:\penalty0 020333, May 2022.
\newblock \doi{10.1103/PRXQuantum.3.020333}.

\bibitem[Zhou et~al.(2020)Zhou, Yang, Hamma, and
  Chamon]{Shiyu2020Entanglemagic}
Shiyu Zhou, Zhi-Cheng Yang, Alioscia Hamma, and Claudio Chamon.
\newblock Single t gate in a clifford circuit drives transition to universal
  entanglement spectrum statistics.
\newblock \emph{SciPost Phys.}, 9:\penalty0 087, 2020.
\newblock \doi{10.21468/SciPostPhys.9.6.087}.

\bibitem[Haferkamp et~al.(2023)Haferkamp, Montealegre-Mora, Heinrich, Eisert,
  Gross, and Roth]{haferkamp2020quantum}
Jonas Haferkamp, Felipe Montealegre-Mora, Markus Heinrich, Jens Eisert, David
  Gross, and Ingo Roth.
\newblock Efficient unitary designs with a system-size independent number of
  non-clifford gates.
\newblock \emph{Communications in Mathematical Physics}, 397:\penalty0 994, Feb
  2023.
\newblock \doi{10.1007/s00220-022-04507-6}.

\bibitem[True and Hamma(2022)]{true2022transitions}
Sarah True and Alioscia Hamma.
\newblock Transitions in {E}ntanglement {C}omplexity in {R}andom {C}ircuits.
\newblock \emph{{Quantum}}, 6:\penalty0 818, September 2022.
\newblock ISSN 2521-327X.
\newblock \doi{10.22331/q-2022-09-22-818}.

\bibitem[Sewell and White(2022)]{Sewell2022Manathermal}
Troy~J. Sewell and Christopher~David White.
\newblock Mana and thermalization: Probing the feasibility of near-clifford
  hamiltonian simulation.
\newblock \emph{Phys. Rev. B}, 106:\penalty0 125130, Sep 2022.
\newblock \doi{10.1103/PhysRevB.106.125130}.

\bibitem[Leone et~al.(2021)Leone, Oliviero, Zhou, and Hamma]{Leone2021Chaos}
Lorenzo Leone, Salvatore F.~E. Oliviero, You Zhou, and Alioscia Hamma.
\newblock Quantum chaos is quantum.
\newblock \emph{Quantum}, 5:\penalty0 453, may 2021.
\newblock \doi{10.22331/q-2021-05-04-453}.

\bibitem[Goto et~al.(2022)Goto, Nosaka, and Nozaki]{goto2022probing}
Kanato Goto, Tomoki Nosaka, and Masahiro Nozaki.
\newblock Probing chaos by magic monotones.
\newblock \emph{Phys. Rev. D}, 106:\penalty0 126009, Dec 2022.
\newblock \doi{10.1103/PhysRevD.106.126009}.

\bibitem[Garcia et~al.(2023)Garcia, Bu, and Jaffe]{garcia2023resource}
Roy~J. Garcia, Kaifeng Bu, and Arthur Jaffe.
\newblock Resource theory of quantum scrambling.
\newblock \emph{Proceedings of the National Academy of Sciences}, 120\penalty0
  (17):\penalty0 e2217031120, 2023.
\newblock \doi{10.1073/pnas.2217031120}.

\bibitem[White et~al.(2021)White, Cao, and Swingle]{White2021CFTmagic}
Christopher~David White, ChunJun Cao, and Brian Swingle.
\newblock Conformal field theories are magical.
\newblock \emph{Phys. Rev. B}, 103:\penalty0 075145, Feb 2021.
\newblock \doi{10.1103/PhysRevB.103.075145}.

\bibitem[Leone et~al.(2024)Leone, Oliviero, Lloyd, and
  Hamma]{leone2022learning}
Lorenzo Leone, Salvatore F.~E. Oliviero, Seth Lloyd, and Alioscia Hamma.
\newblock Learning efficient decoders for quasichaotic quantum scramblers.
\newblock \emph{Phys. Rev. A}, 109:\penalty0 022429, Feb 2024.
\newblock \doi{10.1103/PhysRevA.109.022429}.

\bibitem[Leone et~al.(2022{\natexlab{a}})Leone, Oliviero, Piemontese, True, and
  Hamma]{leone2022retrieving}
Lorenzo Leone, Salvatore F.~E. Oliviero, Stefano Piemontese, Sarah True, and
  Alioscia Hamma.
\newblock Retrieving information from a black hole using quantum machine
  learning.
\newblock \emph{Phys. Rev. A}, 106:\penalty0 062434, Dec 2022{\natexlab{a}}.
\newblock \doi{10.1103/PhysRevA.106.062434}.

\bibitem[Veitch et~al.(2012)Veitch, Ferrie, Gross, and
  Emerson]{veitch2012negative}
Victor Veitch, Christopher Ferrie, David Gross, and Joseph Emerson.
\newblock Negative quasi-probability as a resource for quantum computation.
\newblock \emph{New Journal of Physics}, 14\penalty0 (11):\penalty0 113011,
  2012.
\newblock \doi{10.1088/1367-2630/14/11/113011}.

\bibitem[Seddon and Campbell(2019)]{Seddon2019channel}
James~R. Seddon and Earl~T. Campbell.
\newblock Quantifying magic for multi-qubit operations.
\newblock \emph{Proceedings of the Royal Society A: Mathematical, Physical and
  Engineering Sciences}, 475\penalty0 (2227):\penalty0 20190251, jul 2019.
\newblock \doi{10.1098/rspa.2019.0251}.

\bibitem[Beverland et~al.(2020)Beverland, Campbell, Howard, and
  Kliuchnikov]{Beverland2020NJP}
Michael Beverland, Earl Campbell, Mark Howard, and Vadym Kliuchnikov.
\newblock Lower bounds on the non-clifford resources for quantum computations.
\newblock \emph{Quantum Science and Technology}, 5\penalty0 (3):\penalty0
  035009, may 2020.
\newblock \doi{10.1088/2058-9565/ab8963}.

\bibitem[Wang et~al.(2019)Wang, Wilde, and Su]{wang2019quantifying}
Xin Wang, Mark~M Wilde, and Yuan Su.
\newblock Quantifying the magic of quantum channels.
\newblock \emph{New Journal of Physics}, 21\penalty0 (10):\penalty0 103002,
  2019.
\newblock \doi{10.1088/1367-2630/ab451d}.

\bibitem[Wang et~al.(2020)Wang, Wilde, and Su]{wang2020efficiently}
Xin Wang, Mark~M Wilde, and Yuan Su.
\newblock Efficiently computable bounds for magic state distillation.
\newblock \emph{Physical review letters}, 124\penalty0 (9):\penalty0 090505,
  2020.
\newblock \doi{10.1103/PhysRevLett.124.090505}.

\bibitem[Haug and Kim(2023)]{Haug2022Bellmagic}
Tobias Haug and M.S. Kim.
\newblock Scalable measures of magic resource for quantum computers.
\newblock \emph{{PRX} Quantum}, 4\penalty0 (1), jan 2023.
\newblock \doi{10.1103/prxquantum.4.010301}.

\bibitem[Saxena and Gour(2022)]{saxena2022quantifying}
Gaurav Saxena and Gilad Gour.
\newblock Quantifying multiqubit magic channels with completely
  stabilizer-preserving operations.
\newblock \emph{Phys. Rev. A}, 106:\penalty0 042422, Oct 2022.
\newblock \doi{10.1103/PhysRevA.106.042422}.

\bibitem[Bu et~al.(2023)Bu, Gu, and Jaffe]{bu2023quantum}
Kaifeng Bu, Weichen Gu, and Arthur Jaffe.
\newblock Quantum entropy and central limit theorem.
\newblock \emph{Proceedings of the National Academy of Sciences}, 120\penalty0
  (25):\penalty0 e2304589120, 2023.
\newblock \doi{10.1073/pnas.2304589120}.

\bibitem[Heinrich and Gross(2019)]{Heinrich2019robustnessofmagic}
Markus Heinrich and David Gross.
\newblock Robustness of {M}agic and {S}ymmetries of the {S}tabiliser
  {P}olytope.
\newblock \emph{{Quantum}}, 3:\penalty0 132, April 2019.
\newblock ISSN 2521-327X.
\newblock \doi{10.22331/q-2019-04-08-132}.

\bibitem[Leone et~al.(2022{\natexlab{b}})Leone, Oliviero, and
  Hamma]{Leone2022SRE}
Lorenzo Leone, Salvatore F.~E. Oliviero, and Alioscia Hamma.
\newblock Stabilizer r\'enyi entropy.
\newblock \emph{Phys. Rev. Lett.}, 128:\penalty0 050402, Feb
  2022{\natexlab{b}}.
\newblock \doi{10.1103/PhysRevLett.128.050402}.

\bibitem[Leone and Bittel(2024)]{leone2024stabilizer}
Lorenzo Leone and Lennart Bittel.
\newblock Stabilizer entropies are monotones for magic-state resource theory.
\newblock \emph{arXiv:2404.11652}, 2024.
\newblock \doi{10.48550/arXiv.2404.11652}.

\bibitem[Oliviero et~al.(2022{\natexlab{a}})Oliviero, Leone, Hamma, and
  Lloyd]{oliviero2022measuring}
Salvatore F.~E. Oliviero, Lorenzo Leone, Alioscia Hamma, and Seth Lloyd.
\newblock Measuring magic on a quantum processor.
\newblock \emph{npj Quantum Information}, 8\penalty0 (1):\penalty0 148, Dec
  2022{\natexlab{a}}.
\newblock ISSN 2056-6387.
\newblock \doi{10.1038/s41534-022-00666-5}.

\bibitem[Haug et~al.(2023)Haug, Lee, and Kim]{haug2023efficient}
Tobias Haug, Soovin Lee, and MS~Kim.
\newblock Efficient stabilizer entropies for quantum computers.
\newblock \emph{arXiv preprint arXiv:2305.19152}, 2023.
\newblock \doi{10.48550/arXiv.2305.19152}.

\bibitem[Oliviero et~al.(2022{\natexlab{b}})Oliviero, Leone, and
  Hamma]{oliviero2022magic}
Salvatore F.~E. Oliviero, Lorenzo Leone, and Alioscia Hamma.
\newblock Magic-state resource theory for the ground state of the
  transverse-field ising model.
\newblock \emph{Phys. Rev. A}, 106:\penalty0 042426, Oct 2022{\natexlab{b}}.
\newblock \doi{10.1103/PhysRevA.106.042426}.

\bibitem[Odavić et~al.(2023)Odavić, Haug, Torre, Hamma, Franchini, and
  Giampaolo]{odavic2022complexity}
Jovan Odavić, Tobias Haug, Gianpaolo Torre, Alioscia Hamma, Fabio Franchini,
  and Salvatore~Marco Giampaolo.
\newblock {Complexity of frustration: A new source of non-local
  non-stabilizerness}.
\newblock \emph{SciPost Phys.}, 15:\penalty0 131, 2023.
\newblock \doi{10.21468/SciPostPhys.15.4.131}.

\bibitem[Rattacaso et~al.(2023)Rattacaso, Leone, Oliviero, and
  Hamma]{rattacaso2023stabilizer}
Davide Rattacaso, Lorenzo Leone, Salvatore F.~E. Oliviero, and Alioscia Hamma.
\newblock Stabilizer entropy dynamics after a quantum quench.
\newblock \emph{Phys. Rev. A}, 108:\penalty0 042407, Oct 2023.
\newblock \doi{10.1103/PhysRevA.108.042407}.

\bibitem[Piemontese et~al.(2023)Piemontese, Roscilde, and
  Hamma]{piemontese2023entanglement}
Stefano Piemontese, Tommaso Roscilde, and Alioscia Hamma.
\newblock Entanglement complexity of the rokhsar-kivelson-sign wavefunctions.
\newblock \emph{Phys. Rev. B}, 107:\penalty0 134202, Apr 2023.
\newblock \doi{10.1103/PhysRevB.107.134202}.

\bibitem[Leone et~al.(2023{\natexlab{a}})Leone, Oliviero, and
  Hamma]{leone2023learningTdope}
Lorenzo Leone, Salvatore~FE Oliviero, and Alioscia Hamma.
\newblock Learning t-doped stabilizer states.
\newblock \emph{arXiv preprint arXiv:2305.15398}, 2023{\natexlab{a}}.
\newblock \doi{10.48550/arXiv.2305.15398}.

\bibitem[Leone et~al.(2023{\natexlab{b}})Leone, Oliviero, and
  Hamma]{leone2023nonstabilizerness}
Lorenzo Leone, Salvatore F.~E. Oliviero, and Alioscia Hamma.
\newblock Nonstabilizerness determining the hardness of direct fidelity
  estimation.
\newblock \emph{Phys. Rev. A}, 107:\penalty0 022429, Feb 2023{\natexlab{b}}.
\newblock \doi{10.1103/PhysRevA.107.022429}.

\bibitem[Haug and Piroli(2023{\natexlab{a}})]{Haug2022MPSmagic}
Tobias Haug and Lorenzo Piroli.
\newblock Quantifying nonstabilizerness of matrix product states.
\newblock \emph{Phys. Rev. B}, 107:\penalty0 035148, Jan 2023{\natexlab{a}}.
\newblock \doi{10.1103/PhysRevB.107.035148}.

\bibitem[Lami and Collura(2023)]{lami2023quantum}
Guglielmo Lami and Mario Collura.
\newblock Quantum magic via perfect pauli sampling of matrix product states.
\newblock \emph{arXiv:2303.05536}, 2023.
\newblock \doi{10.48550/arXiv.2303.05536}.

\bibitem[Haug and Piroli(2023{\natexlab{b}})]{haug2023stabilizer}
Tobias Haug and Lorenzo Piroli.
\newblock Stabilizer entropies and nonstabilizerness monotones.
\newblock \emph{{Quantum}}, 7:\penalty0 1092, August 2023{\natexlab{b}}.
\newblock ISSN 2521-327X.
\newblock \doi{10.22331/q-2023-08-28-1092}.

\bibitem[Tarabunga et~al.(2023)Tarabunga, Tirrito, Chanda, and
  Dalmonte]{tarabunga2023many}
Poetri~Sonya Tarabunga, Emanuele Tirrito, Titas Chanda, and Marcello Dalmonte.
\newblock Many-body magic via pauli-markov chains---from criticality to gauge
  theories.
\newblock \emph{PRX Quantum}, 4:\penalty0 040317, Oct 2023.
\newblock \doi{10.1103/PRXQuantum.4.040317}.

\bibitem[Rossi et~al.(2013)Rossi, Huber, Bruß, and
  Macchiavello]{Rossi2013Hyper}
M~Rossi, M~Huber, D~Bruß, and C~Macchiavello.
\newblock Quantum hypergraph states.
\newblock \emph{New Journal of Physics}, 15\penalty0 (11):\penalty0 113022, nov
  2013.
\newblock \doi{10.1088/1367-2630/15/11/113022}.

\bibitem[Qu et~al.(2013)Qu, Wang, Li, and Bao]{Qu2013hypergraph}
Ri~Qu, Juan Wang, Zong-shang Li, and Yan-ru Bao.
\newblock Encoding hypergraphs into quantum states.
\newblock \emph{Phys. Rev. A}, 87:\penalty0 022311, Feb 2013.
\newblock \doi{10.1103/PhysRevA.87.022311}.

\bibitem[Raussendorf and Briegel(2001)]{Raussendorf2001onewayQC}
Robert Raussendorf and Hans~J. Briegel.
\newblock A one-way quantum computer.
\newblock \emph{Phys. Rev. Lett.}, 86:\penalty0 5188--5191, May 2001.
\newblock \doi{10.1103/PhysRevLett.86.5188}.

\bibitem[Raussendorf et~al.(2003)Raussendorf, Browne, and
  Briegel]{Raussendorf2003MBQC}
Robert Raussendorf, Daniel~E. Browne, and Hans~J. Briegel.
\newblock Measurement-based quantum computation on cluster states.
\newblock \emph{Phys. Rev. A}, 68:\penalty0 022312, Aug 2003.
\newblock \doi{10.1103/PhysRevA.68.022312}.

\bibitem[Bremner et~al.(2016)Bremner, Montanaro, and Shepherd]{Bremner2016IQP}
Michael~J. Bremner, Ashley Montanaro, and Dan~J. Shepherd.
\newblock Average-case complexity versus approximate simulation of commuting
  quantum computations.
\newblock \emph{Phys. Rev. Lett.}, 117:\penalty0 080501, Aug 2016.
\newblock \doi{10.1103/PhysRevLett.117.080501}.

\bibitem[Miller and Miyake(2016)]{Miller2016Hierarchy}
Jacob Miller and Akimasa Miyake.
\newblock Hierarchy of universal entanglement in 2d measurement-based quantum
  computation.
\newblock \emph{npj Quantum Information}, 2\penalty0 (1):\penalty0 1--6, 2016.
\newblock \doi{10.1038/npjqi.2016.36}.

\bibitem[Takeuchi et~al.(2019)Takeuchi, Morimae, and
  Hayashi]{Takeuchi2019PauliMBQC}
Yuki Takeuchi, Tomoyuki Morimae, and Masahito Hayashi.
\newblock Quantum computational universality of hypergraph states with pauli-x
  and z basis measurements.
\newblock \emph{Scientific Reports}, 9\penalty0 (1):\penalty0 13585, Sep 2019.
\newblock ISSN 2045-2322.
\newblock \doi{10.1038/s41598-019-49968-3}.

\bibitem[Levin and Gu(2012)]{Levin2012Braiding}
Michael Levin and Zheng-Cheng Gu.
\newblock Braiding statistics approach to symmetry-protected topological
  phases.
\newblock \emph{Phys. Rev. B}, 86:\penalty0 115109, Sep 2012.
\newblock \doi{10.1103/PhysRevB.86.115109}.

\bibitem[Yoshida(2016)]{Yoshida2016Topological}
Beni Yoshida.
\newblock Topological phases with generalized global symmetries.
\newblock \emph{Phys. Rev. B}, 93:\penalty0 155131, Apr 2016.
\newblock \doi{10.1103/PhysRevB.93.155131}.

\bibitem[Miller and Miyake(2018)]{Miller2018Latent}
Jacob Miller and Akimasa Miyake.
\newblock Latent computational complexity of symmetry-protected topological
  order with fractional symmetry.
\newblock \emph{Phys. Rev. Lett.}, 120:\penalty0 170503, Apr 2018.
\newblock \doi{10.1103/PhysRevLett.120.170503}.

\bibitem[White and Wilson(2020)]{white2020mana}
Christopher~David White and Justin~H Wilson.
\newblock Mana in haar-random states.
\newblock \emph{arXiv preprint arXiv:2011.13937}, 2020.
\newblock \doi{10.48550/arXiv.2011.13937}.

\bibitem[Hein et~al.(2004)Hein, Eisert, and Briegel]{Hein2004Multiparty}
M.~Hein, J.~Eisert, and H.~J. Briegel.
\newblock Multiparty entanglement in graph states.
\newblock \emph{Phys. Rev. A}, 69:\penalty0 062311, Jun 2004.
\newblock \doi{10.1103/PhysRevA.69.062311}.

\bibitem[Bell et~al.(2014)Bell, Herrera-Mart{\'\i}, Tame, Markham, Wadsworth,
  and Rarity]{bell2014experimental}
Bryn~A Bell, DA~Herrera-Mart{\'\i}, MS~Tame, Damian Markham, WJ~Wadsworth, and
  JG~Rarity.
\newblock Experimental demonstration of a graph state quantum error-correction
  code.
\newblock \emph{Nature communications}, 5\penalty0 (1):\penalty0 3658, 2014.
\newblock \doi{10.1038/ncomms4658}.

\bibitem[Kruszynska and Kraus(2009{\natexlab{a}})]{Kruszynska2009Hyper}
C.~Kruszynska and B.~Kraus.
\newblock Local entanglability and multipartite entanglement.
\newblock \emph{Phys. Rev. A}, 79:\penalty0 052304, May 2009{\natexlab{a}}.
\newblock \doi{10.1103/PhysRevA.79.052304}.

\bibitem[T\'oth and G\"uhne(2005)]{toth2005detecting}
G\'eza T\'oth and Otfried G\"uhne.
\newblock Detecting genuine multipartite entanglement with two local
  measurements.
\newblock \emph{Phys. Rev. Lett.}, 94:\penalty0 060501, Feb 2005.
\newblock \doi{10.1103/PhysRevLett.94.060501}.

\bibitem[G{\"u}hne et~al.(2014)G{\"u}hne, Cuquet, Steinhoff, Moroder, Rossi,
  Bru{\ss}, Kraus, and Macchiavello]{guhne2014hypergraph}
Otfried G{\"u}hne, Marti Cuquet, Frank~ES Steinhoff, Tobias Moroder, Matteo
  Rossi, Dagmar Bru{\ss}, Barbara Kraus, and Chiara Macchiavello.
\newblock Entanglement and nonclassical properties of hypergraph states.
\newblock \emph{Journal of Physics A: Mathematical and Theoretical},
  47\penalty0 (33):\penalty0 335303, 2014.
\newblock \doi{10.1088/1751-8113/47/33/335303}.

\bibitem[Zhou and Hamma(2022)]{zhou2022hyper}
You Zhou and Alioscia Hamma.
\newblock Entanglement of random hypergraph states.
\newblock \emph{Phys. Rev. A}, 106:\penalty0 012410, Jul 2022.
\newblock \doi{10.1103/PhysRevA.106.012410}.

\bibitem[Tomamichel(2012)]{tomamichel2012framework}
Marco Tomamichel.
\newblock A framework for non-asymptotic quantum information theory.
\newblock \emph{arXiv preprint arXiv:1203.2142}, 2012.
\newblock \doi{10.48550/arXiv.1203.2142}.

\bibitem[Liu et~al.(2019)Liu, Bu, and Takagi]{Liu2019OneShot}
Zi-Wen Liu, Kaifeng Bu, and Ryuji Takagi.
\newblock One-shot operational quantum resource theory.
\newblock \emph{Phys. Rev. Lett.}, 123:\penalty0 020401, Jul 2019.
\newblock \doi{10.1103/PhysRevLett.123.020401}.

\bibitem[Kruszynska and Kraus(2009{\natexlab{b}})]{Kruszynska2009Local}
C.~Kruszynska and B.~Kraus.
\newblock Local entanglability and multipartite entanglement.
\newblock \emph{Phys. Rev. A}, 79:\penalty0 052304, May 2009{\natexlab{b}}.
\newblock \doi{10.1103/PhysRevA.79.052304}.

\bibitem[Nakata and Murao(2014)]{Nakata2014review}
Yoshifumi Nakata and Mio Murao.
\newblock Diagonal quantum circuits: Their computational power and
  applications.
\newblock \emph{The European Physical Journal Plus}, 129\penalty0 (7):\penalty0
  152, Jul 2014.
\newblock ISSN 2190-5444.
\newblock \doi{10.1140/epjp/i2014-14152-9}.

\bibitem[Iaconis(2021)]{iaconis2021quantum}
Jason Iaconis.
\newblock Quantum state complexity in computationally tractable quantum
  circuits.
\newblock \emph{PRX Quantum}, 2:\penalty0 010329, Feb 2021.
\newblock \doi{10.1103/PRXQuantum.2.010329}.

\bibitem[Nakata and Murao(2020)]{nakata2020generic}
Yoshifumi Nakata and Mio Murao.
\newblock Generic entanglement entropy for quantum states with symmetry.
\newblock \emph{Entropy}, 22\penalty0 (6):\penalty0 684, 2020.
\newblock \doi{10.3390/e22060684}.

\bibitem[Li and Haldane(2008)]{Haldane2008Spectrum}
Hui Li and F.~D.~M. Haldane.
\newblock Entanglement spectrum as a generalization of entanglement entropy:
  Identification of topological order in non-abelian fractional quantum hall
  effect states.
\newblock \emph{Phys. Rev. Lett.}, 101:\penalty0 010504, Jul 2008.
\newblock \doi{10.1103/PhysRevLett.101.010504}.

\bibitem[Chamon et~al.(2014)Chamon, Hamma, and
  Mucciolo]{Chamon2014Irreversibility}
Claudio Chamon, Alioscia Hamma, and Eduardo~R. Mucciolo.
\newblock Emergent irreversibility and entanglement spectrum statistics.
\newblock \emph{Phys. Rev. Lett.}, 112:\penalty0 240501, Jun 2014.
\newblock \doi{10.1103/PhysRevLett.112.240501}.

\bibitem[Shaffer et~al.(2014)Shaffer, Chamon, Hamma, and
  Mucciolo]{shaffer2014irreversibility}
Daniel Shaffer, Claudio Chamon, Alioscia Hamma, and Eduardo~R Mucciolo.
\newblock Irreversibility and entanglement spectrum statistics in quantum
  circuits.
\newblock \emph{Journal of Statistical Mechanics: Theory and Experiment},
  2014\penalty0 (12):\penalty0 P12007, 2014.
\newblock \doi{10.1088/1742-5468/2014/12/P12007}.

\bibitem[Bao et~al.(2022)Bao, Cao, and Su]{bao2022magic}
Ning Bao, ChunJun Cao, and Vincent~Paul Su.
\newblock Magic state distillation from entangled states.
\newblock \emph{Physical Review A}, 105\penalty0 (2):\penalty0 022602, 2022.
\newblock \doi{10.1103/PhysRevA.105.022602}.

\bibitem[Tirrito et~al.(2023)Tirrito, Tarabunga, Lami, Chanda, Leone, Oliviero,
  Dalmonte, Collura, and Hamma]{tirrito2023quantifying}
Emanuele Tirrito, Poetri~Sonya Tarabunga, Gugliemo Lami, Titas Chanda, Lorenzo
  Leone, Salvatore~FE Oliviero, Marcello Dalmonte, Mario Collura, and Alioscia
  Hamma.
\newblock Quantifying non-stabilizerness through entanglement spectrum
  flatness.
\newblock \emph{arXiv preprint arXiv:2304.01175}, 2023.
\newblock \doi{10.48550/arXiv.2304.01175}.

\end{thebibliography}

\newpage



\begin{appendix}
\section{Related Proofs of the general formula of magic for hypergraph states}
\subsection{Proof of Observation \ref{ob:beauty}}\label{App:beauty}
We first prove some commuting relation that 
\begin{equation}\label{eq:induction}
    \left(\prod_{v_i\in e'} X_i\right)CZ_e\left(\prod_{v_i\in e'} X_i\right)=\prod_{q\subseteq e'}CZ_{e\setminus q}
\end{equation}
for any set $e'\subseteq e$ by induction. When there is only one element in $e'$, it is just the formula $X_iCZ_eX_i=CZ_e\cdot CZ_{e\setminus\{v_i\}}$. Suppose for some $e'$ such that Eq.~\eqref{eq:induction} holds, then for any $e''=e'\cup\{v_{i_0}\}$ with $e'\subset e''\subseteq e$, one has
\begin{equation}
\begin{aligned}
    \left(\prod_{v_i\in e''} X_i\right)CZ_e\left(\prod_{v_i\in e''} X_i\right)&=X_{i_0}\left(\left(\prod_{v_i\in e'} X_i\right)CZ_e\left(\prod_{v_i\in e'} X_i\right)\right)X_{i_0}\\
    &=X_{i_0}\left(\prod_{q\subseteq e'}CZ_{e\setminus q}\right)X_{i_0}\\
    &=\left(\prod_{q\subseteq e'}CZ_{e\setminus q}\right)\cdot\left(\prod_{q\subseteq e'}CZ_{e\setminus (q\cup\{v_{i_0}\})}\right)\\
    &=\left(\prod_{q\subseteq e''}CZ_{e\setminus q}\right),
\end{aligned}
\end{equation}
where the second line is by the induction assumption, and the third line is by the single-element case. Consequently, Eq.~\eqref{eq:induction} holds for any set $e'\subseteq e$. 

It is clear that the stabilizer generators of the initial state $\ket{\Psi_0}=\ket{+}^{\otimes n}$ are $X_i$. It is direct to check that $S_i=U(G)X_iU^{\dag}(G)$. Then we calculate the product of stabilizer generators as
\begin{equation}
\begin{aligned}
    \mathrm{St}(G,\vec{s})&:=\prod_{i}S_i^{s_i}=U(G)\cdot\prod_{i} X_i^{s_i}\cdot U^{\dag}(G)\\
    &=\prod_{i} X_i^{s_i}\cdot \prod_{e\in E}\left(\left(\prod_{i} X_i^{s_i}\right)CZ_e\left(\prod_{i} X_i^{s_i}\right)\right)\cdot\prod_{e\in E}CZ_e\\
    &=\prod_{i} X_i^{s_i}\cdot \prod_{e\in E}\left(\left(\prod_{v_i\in (v(\vec{s})\cap e)} X_i\right)CZ_e\left(\prod_{v_i\in (v(\vec{s})\cap e)} X_i\right)\right)\cdot\prod_{e\in E}CZ_e\\
    &=\prod_{i}X_i^{s_i}\cdot\prod_{e\in E}\prod_{q\subseteq (v(\vec{s})\cap e)}CZ_{e\setminus q}\cdot\prod_{e\in E}CZ_e\\
    &=\prod_{i}X_i^{s_i}\cdot\prod_{e\in E}\prod_{q\neq \emptyset,q\subseteq (v(\vec{s})\cap e)}CZ_{e\setminus q}.
\end{aligned}
\end{equation}
Here the second line we insert many additional $\prod_{i} X_i^{s_i}$ for each $CZ_e$. Note that in the third line, nontrivial contributions only come from $v_i\in (v(\vec{s})\cap e)$, otherwise $X_i$ commutes with $CZ_e$. The fourth line is by Eq.~\eqref{eq:induction}.

\subsection{Proof of Proposition \ref{prop:PLcomp}}\label{App:PLcomp}
To proceed with the proof, we first show the following simple but useful result about the action of multi-qubit Pauli $X$ gates with any phase unitary $U^Z=e^{i\theta(\vec{a})}\kb{\vec{a}}$, which only introduces phase on the computational basis $\{\ket{\vec{a}}\}$ for some function $\theta$ of the n-bit string $\vec{a}$.
\begin{lemma}\label{le:1}
    Let $X^{\vec{x}}=\bigotimes_i X_i^{x_i}$ with $X_i$ the Pauli $X$ gate on qubit $i$. If $U^Z$ is a phase unitary, one has
    \begin{equation}
	\Tr{X^{\vec{x}}U^Z}=\delta_{\vec{x},\vec{0}}\Tr{U^Z}.
    \end{equation}
    Moreover, if $U_1^Z$ and $U_2^Z$ are both phase unitary, 
    \begin{equation}
	\Tr{X^{\vec{x}_2}U_2^Z X^{\vec{x}_1}U_1^Z}=\delta_{\vec{x}_1,\vec{x}_2}\Tr{X^{\vec{x}_2}U_2^Z X^{\vec{x}_1}U_1^Z}.
    \end{equation}
\end{lemma}
\begin{proof}
The proof is straightforward by definition.
\begin{equation}
    \Tr{X^{\vec{x}}U^Z}=\sum_{\vec{a}}\bra{\vec{a}}X^{\vec{x}}U^Z\ket{\vec{a}}=\sum_{\vec{a}}\bra{\vec{a}}X^{\vec{x}}\ket{\vec{a}}e^{i\theta(\vec{a})}=\delta_{\vec{x},\vec{0}}\Tr{U_Z},
\end{equation}
where $U^Z$ only introduces some phase $\theta(\vec{a})$ on $\ket{\vec{a}}$, and 
\begin{equation}
\begin{aligned}
    \Tr{X^{\vec{x}_2}U_2^Z X^{\vec{x}_1}U_1^Z}
    =&\sum_{\vec{a}}\bra{\vec{a}}X^{\vec{x}_2}U_2^Z X^{\vec{x}_1}U_1^Z\ket{\vec{a}}\\
    =&\sum_{\vec{a}}\bra{\vec{a}+\vec{x}_2}U_2^Z\ket{\vec{a}+\vec{x}_1}e^{i\theta_1(\vec{a})}\\
    =&\delta_{\vec{x}_1,\vec{x}_2}\Tr{X^{\vec{x}_2}U_2^Z X^{\vec{x}_1}U_1^Z}.
\end{aligned}
\end{equation}
\end{proof}

It is clear that $CZ_{e}$ gates used to generate hypergraph states in Eq.~\eqref{eq:def:hyper} are phase gates, and we can apply Lemma \ref{le:1} to calculate its PL-component as follows. By using the stabilizer decomposition of $\ket{G}$ in  Eq.~\eqref{eq:stabG}, \eqref{eq:stabDecom} and \eqref{eq:beauty}, and the definition of Pauli operator in Eq.~\eqref{eq:pauliP}, one has
\begin{equation}\label{eq:Xloc2stab}
\begin{aligned}
    \Tr{P_{\vec{x},\vec{z}}\kb{G}}=&2^{-n}\Tr{P_{\vec{x},\vec{z}}\ \sum_{\vec{s}} \mathrm{St}(G,\vec{s})}\\
    =&2^{-n}\omega(\vec{x},\vec{z}) \sum_{\vec{s}}\Tr{X^{\vec{x}} Z^{\vec{z}}\cdot X^{\vec{s}}\prod_{e}\prod_{q\subseteq (v(\vec{s})\cap e)}CZ_{e\setminus q}}\\
    =&2^{-n}\omega(\vec{x},\vec{z}) \Tr{X^{\vec{x}} Z^{\vec{z}} X^{\vec{x}}\prod_{e}\prod_{q\subseteq (v(\vec{x})\cap e)}CZ_{e\setminus q}}\\
    =&2^{-n}\omega^3(\vec{x},\vec{z})\Tr{\prod_i Z_i^{z_i}\cdot\prod_{e}\prod_{q\subseteq (v(\vec{x})\cap e)}CZ_{e\setminus q}},
\end{aligned}
\end{equation}
where the third line is by Lemma \ref{le:1} inducing the constraint $\delta_{\vec{x},\vec{s}}$, and the final line is  by the commuting relation $X^{\vec{x}} Z^{\vec{z}} X^{\vec{x}}=\omega^2(\vec{x},\vec{z})Z^{\vec{z}}$.
Eq.~\eqref{eq:Xloc2stab} indicates that one should only select the stabilizer generator of $\prod_i  S_i^{s_i}$ in the summation \emph{exactly} according to $\vec{x}$ information of the Pauli operator $P_{\vec{x},\vec{z}}$. 

The operator in the trace of the final line of Eq.~\eqref{eq:Xloc2stab} is indeed a phase unitary denoted as
\begin{equation}\label{eq:UGin}
\begin{aligned}
    \tilde{U}(G_{\vec{x},\vec{z}})&:=\prod_i Z_i^{z_i}\cdot\prod_{e}\prod_{q\subseteq (v(\vec{x})\cap e)}CZ_{e\setminus q}\\
    &=(-1)^{\nu(\vec{x})}\prod_i Z_i^{z_i}\cdot\prod_{e}\prod_{q\subseteq (v(\vec{x})\cap e),q\neq e}CZ_{e\setminus q}\\
    &:=(-1)^{\nu(\vec{x})}U(G_{\vec{x},\vec{z}})
\end{aligned}
\end{equation}
which can be represented by a hypergraph $G_{\vec{x},\vec{z}}$ defined in Eq.~\eqref{eq:Gin} by considering $1$-hyperedges and other hyperedges separately, similar as in the definition of hypergraph state in Eq.~\eqref{eq:def:hyper}. Notice that $\nu(\vec{x})$ is the number of $e\subseteq v(\vec{x})$, i.e., it counts the number of $q=e$. In this case, $CZ_{e\setminus q}=CZ_{\emptyset}$ which gives a $-1$ phase by definition in the second line. Since both $\Tr{P_{\vec{x},\vec{z}}\kb{G}}$ and $\Tr{U(G_{\vec{x},\vec{z}})}$ are real numbers, the phases are not essential and by Eq.~\eqref{eq:Xloc2stab} we finally have 
\begin{equation}
\begin{aligned}
    \Tr{P_{\vec{x},\vec{z}}\kb{G}}^2=2^{-2n}\Tr{U(G_{\vec{x},\vec{z}})}^2.
\end{aligned}
\end{equation}

\section{Proof of Theorem \ref{th:bdegree}: upper bound of magic from bounded average degree}\label{App:Thdegree}
First, we find the subset $\tilde{V}_{\vec{x}}$ by deleting zero-degree vertices in $V$. That is, all the vertices now has some neighbour,  $\tilde{V}_{\vec{x}}=\{v_i\in V | \exists\ e\in E^{(2)}_{\vec{x}}, v_i\in e\}$, which is only determined by $E^{(2)}_{\vec{x}}$ and thus $\vec{x}$. In this way, for each induced hypergraph $G_{\vec{x},\vec{z}}^{*}$, we define the reduced hypergraph $\tilde{G}_{\vec{x},\vec{z}}^{*}=\{\tilde{V}_{\vec{x}},\tilde{E}_{\vec{x},\vec{z}}^{*}\}$, and here $\tilde{E}_{\vec{x},\vec{z}}^{*}$ is just maintained by reducing $E_{\vec{x},\vec{z}}^{*}$ on the subset $\tilde{V}_{\vec{x}}$. Actually, this step reduces all the vertices that are not connected to other vertices and delete all the related $1$-hyperedges in $E^{(1)}_{\vec{z}}$. 

\begin{figure}[tbhp!]
\centering
\resizebox{4cm}{!}{\includegraphics[scale=0.8]{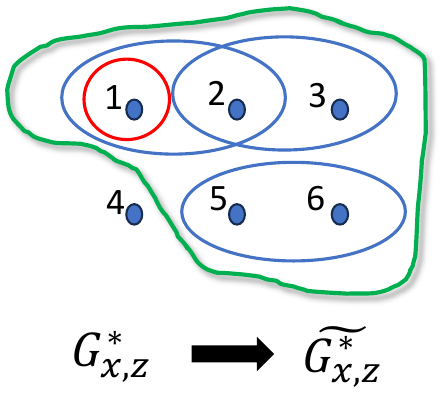}}
\caption{The reduced hypergraph $\tilde{G}_{\vec{x},\vec{z}}^{*}$ from the hypergraph $G_{\vec{x},\vec{z}}^{*}$. Here the vertex $4$ has no neighbour, or is zero-degree, and one has $\tilde{V}_{\vec{x}}=\{1,2,3,5,6\}$. The vertex $4$ is deleted from $G_{\vec{x},\vec{z}}^{*}$ to get $\tilde{G}_{\vec{x},\vec{z}}^{*}$, which is highlighted by the shaded green circle.}  \label{fig:reducedG} 
\end{figure}

If there is some vertex $v_j \in V\setminus \tilde{V}_{\vec{x}}$ with a $1$-hyperedge on it, i.e., there is some Pauli $Z$ operator outside $\tilde{V}_{\vec{x}}$, one has
\begin{equation}
\begin{aligned}
    2^{-n}\Tr{U(G_{\vec{x},\vec{z}}^{*})}=&2^{-n}\Tr{U(\tilde{G}_{\vec{x},\vec{z}}^{*})\otimes\cdots\otimes Z\otimes\cdots}=0.
\end{aligned}
\end{equation}
As a result, we should choose $\vec{z} \in \zeta_{\vec{x}}:=\{\vec{z} | z_i=0, \forall v_i\in V\setminus\tilde{V}_{\vec{x}}\}$ to get nonzero result. That is, there should be no $1$-hyperedge on $V\setminus \tilde{V}_{\vec{x}}$. Denote $m(\vec{x})=|\tilde{V}_{\vec{x}}|$, in this case one has
\begin{equation}
\begin{aligned}
    2^{-n}\Tr{U(G_{\vec{x},\vec{z}}^{*})}=&2^{-n}\Tr{U(\tilde{G}_{\vec{x},\vec{z}}^{*})\otimes \id^{\otimes(n-{m(\vec{x})})}}\\
    =&\frac{1}{2^{m(\vec{x})}}\Tr{U(\tilde{G}_{\vec{x},\vec{z}}^{*})}.
\end{aligned}
\end{equation}

Now by fixing $\vec{x}$ and summarizing all nontrivial vectors $\vec{z}\in \zeta_{\vec{x}}$,
\begin{equation}
\begin{aligned}
    \sum_{\vec{z}}\Tr{U(\tilde{G}_{\vec{x},\vec{z}}^{*})}
    =&\Tr{\sum_{\vec{z} \in \zeta_{\vec{x}}}\prod_j Z_j^{z_j}\prod_{e'\in E^{(2)}_{\vec{x}}}CZ_{e'}}\\
    =&\sum_{\vec{a}}\bra{\vec{a}}\left(\prod_j (\id+Z_j)\prod_{e'\in E^{(2)}_{\vec{x}}}CZ_{e'}\right)\ket{\vec{a}}\\
    =&\sum_{\vec{a}}\delta_{\vec{a},\vec{0}}\cdot2^{m(\vec{x})}\bra{\vec{a}}\prod_{e'\in E^{(2)}_{\vec{x}}}CZ_{e'}\ket{\vec{a}}\\
    =&2^{m(\vec{x})}\bra{\vec{0}}\prod_{e'\in E^{(2)}_{\vec{x}}}CZ_{e'}\ket{\vec{0}}=2^{m(\vec{x})}.
\end{aligned}
\end{equation}

This indicates that $\left|\sum_{\vec{z}}\frac{1}{2^{m(\vec{x})}}\Tr{U(\tilde{G}_{\vec{x},\vec{z}}^{*})}\right|=1$, which means that for $\alpha\geq 2$,
\begin{equation}\label{eq:loose}
    \sum_{\vec{z}}\left|\frac{1}{2^{m(\vec{x})}}\Tr{U(\tilde{G}_{\vec{x},\vec{z}}^{*})}\right|^{2\alpha}\geq 2^{m(\vec{x})}\cdot\left(\frac{1}{2^{m(\vec{x})}}\right)^{2\alpha}=\frac{1}{2^{(2\alpha-1)m(\vec{x})}}.
\end{equation}

Let $W(\vec{x})=w$ be the weight of $\vec{x}$, i.e., the number of $x_i=1$ in $\vec{x}$. Since $m(\vec{x})$ is the number of vertices in edges in $E^{(2)}_{\vec{x}}$, it is bounded by the number of vertices adjacent to those $v_i$ with $x_i=1$. That is, $m(\vec{x})\leq \Delta(\vec{x}):=\sum_{i:x_i=1} \Delta(v_i)$ where $\Delta(v_i)$ is for the degree of vertex $v_i$. 
\begin{equation}
\begin{aligned}
    \mathbf{m}_{\alpha}(\ket{G})=&2^{-n}\sum_{\vec{x}}\left(\sum_{\vec{z}}\left(2^{-n}\Tr{U(G_{\vec{x},\vec{z}}^{*})}\right)^{2\alpha}\right)\\
    \geq&2^{-n}\sum_{\vec{x}} \frac{1}{2^{(2\alpha-1)m(\vec{x})}}\geq 2^{-n}\sum_{\vec{x}} \frac{1}{2^{(2\alpha-1)\Delta(\vec{x})}}\\
    =&2^{-n}\sum_{w=0}^n\sum_{\vec{x}:W(\vec{x})=w}\frac{1}{2^{(2\alpha-1)\Delta(\vec{x})}}\geq 2^{-n}\sum_{w=0}^n{n \choose w}\frac{1}{2^{(2\alpha-1)\overline{\Delta(\vec{x})}^w}}\\
    =&2^{-n}\sum_{w=0}^n{n \choose w}\frac{1}{2^{(2\alpha-1)w\bar{\Delta}(G)}}=2^{-n}\left(1+\frac{1}{2^{(2\alpha-1)\bar{\Delta}(G)}}\right)^n.
\end{aligned}
\end{equation}
In the third line, we apply the convex property of function $f(z)=\frac{1}{2^{(2\alpha-1)z}}$, and the relation
\begin{equation}
    \overline{\Delta(\vec{x})}^{w}:=\frac{1}{{n \choose w}}\sum_{\vec{x}:W(\vec{x})=w}\Delta(\vec{x})=w\bar{\Delta}(G),
\end{equation}
where the degree of each vertex has been summed totally $w$ times. The final line is by the binomial formula.

Finally, the \renyiaa\ is bounded by 
\begin{equation}
    \mb{M}_{\alpha}(\ket{G})=-\frac{1}{\alpha-1}\log \mb{m}_{\alpha}(\ket{G})\leq \frac{1}{\alpha-1}\left(1-\log_2\left(1+\frac{1}{2^{(2\alpha-1)\bar{\Delta}(G)}}\right)\right)n
\end{equation}
for $\alpha\geq 2$.

\section{Proofs of general formulas for statistical properties of PL-moment}

\subsection{Proof of Theorem \ref{theo:randomave}}\label{App:Thmain-counting}
To distinguish among different $c$-uniform hypergraphs, we assign a binary vector $\vec{B}=\left(\cdots,B_{e_c},\cdots\right)$ with ${e_c\subseteq V,|e_c|=c}$ as a label, i.e., $G_{n,c}^{(\vec{B})}$. The value of $B_{e_c}$ decides whether the hyperedge $e_c$ belongs to the edge set of $G_{n,c}^{(\vec{B})}$.

The first step is to write $\mathbf{m}_{\alpha}$ on $2\alpha$-replica \cite{Leone2022SRE,Haug2022MPSmagic},
\begin{equation}\label{eq:base}
\begin{aligned}
    \mathbf{m}_{\alpha}(\ket{\Psi})&=2^{-n}\sum_{P\in\mathcal{P}_n}\left(\Tr{P\ket{\Psi}\bra{\ket{\Psi}}}\right)^{2\alpha}\\
    &=\left(\bra{\Psi}\otimes\bra{\Psi^*}\right)^{\otimes\alpha}\bigotimes_i \Lambda_i^{(\alpha)}\left(\ket{\Psi}\otimes\ket{\Psi^*}\right)^{\otimes\alpha}\\
    &=\Tr{\rho^{(\alpha)}(\Psi)\cdot\bigotimes_i \Lambda_i^{(\alpha)}}
\end{aligned}
\end{equation}
where the local term $\Lambda_i^{(\alpha)}$ is for the $2\alpha$-replica of $i$-th qubit 
\begin{equation}
\begin{aligned}
    \Lambda_i^{(\alpha)}
    =\frac{1}{2}\left(\id_i^{\otimes 2\alpha}+X_i^{\otimes 2\alpha}\right)\left(\id_i^{\otimes 2\alpha}+Z_i^{\otimes 2\alpha}\right),
\end{aligned}
\end{equation}
and
\begin{equation}
    \rho^{(\alpha)}(\Psi)=(\ket{\Psi}\otimes\ket{\Psi^*})^{\otimes\alpha}(\bra{\Psi}\otimes\bra{\Psi^*})^{\otimes\alpha}=(\rho\otimes\rho^*)^{\otimes\alpha},
\end{equation}
where $\rho=\rho^*$ for hypergraph states.
In this step, the $2\alpha$-degree formula is transformed into a $2\alpha$-replica tensor, which enables the interchange of the order of the summation of Pauli operators and the trace. Similarly, we can also interchange the order in the calculation of the average PL-moment
\begin{equation}\label{eq:avemagic}
\begin{aligned}
    \langle\mathbf{m}_{\alpha}\rangle_{\mc{E}_c}&=\frac{1}{2^{{n \choose c}}}\sum_{G_{n,c}\in\mathcal{G}_{n,c}}\mathbf{m}_\alpha(\ket{G_{n,c}})\\
    &=\frac{1}{2^{{n \choose c}}}\sum_{\vec{B}}\Tr{\rho^{(\alpha)}(G_{n,c}^{(\vec{B})})\cdot\bigotimes_i \Lambda_i^{(\alpha)}}\\
    &=\Tr{\frac{1}{2^{{n \choose c}}}\sum_{\vec{B}}\rho^{(\alpha)}(G_{n,c}^{(\vec{B})})\cdot\bigotimes_i \Lambda_i^{(\alpha)}},
\end{aligned}
\end{equation}
where the second line is by Eq.~\eqref{eq:base}. The average of the $\alpha$-fold density matrix in the final of Eq.~\eqref{eq:avemagic} shows
\begin{equation}\label{eq:averho}
\begin{aligned}
    \overline{\rho^{(\alpha)}}&:=\frac{1}{2^{{n \choose c}}}\sum_{\vec{B}}\rho^{(\alpha)}(G_{n,c}^{(\vec{B})})\\
    &=\frac{1}{2^{{n \choose c}}}\cdot\frac{1}{2^{2\alpha n}}\sum_{\vec{B}}\bigotimes_{r=1}^{2\alpha}\sum_{\vec{x}^{(r)}}\mathrm{St}(G_{n,c}^{(\vec{B})},\vec{x}^{(r)})\\
    &=\frac{1}{2^{{n \choose c}}}\cdot\frac{1}{2^{2\alpha n}}\sum_{\vec{x}^{(1)}}\cdots\sum_{\vec{x}^{(2\alpha)}}\sum_{\vec{B}}\prod_{r=1}^{2\alpha}\mathrm{St}(G_{n,c}^{(\vec{B})},\vec{x}^{(r)}),
\end{aligned}
\end{equation}
where in the second line we insert the stabilizer decomposition of the hypergraph state. 
Here the index $r$ labels all $2\alpha$-replica, and the binary vector $\vec{x}^{(r)}$ labels the choice of $S_i$ of each qubit for the $r$-th replica. Without causing any ambiguity, we use $\mathrm{St}(G_{n,c}^{(\vec{B})},\vec{x}^{(r)})$ to replace the complicated expression $\id\otimes\cdots\otimes \mathrm{St}(G_{n,c}^{(\vec{B})},\vec{x}^{(r)})\otimes\cdots\otimes\id$ for $2\alpha$-replica and change the tensor product to the matrix product. Similar simplifications will also be used in the following discussion without further explanations. By using Observation \ref{ob:beauty}, we further get
\begin{equation}\label{eq:Halpha}
\begin{aligned}
    \frac{1}{2^{{n \choose c}}}\sum_{\vec{B}}\prod_{r=1}^{2\alpha}\mathrm{St}(G_{n,c}^{(\vec{B})},\vec{x}^{(r)})
    =& \frac{1}{2^{{n \choose c}}}\sum_{\vec{B}}\prod_{r=1}^{2\alpha}\left(\prod_{i}\left(X_i^{(r)}\right)^{x_i^{(r)}}\prod_{e_c}\prod_{q\subseteq v(\vec{x}^{(r)})\cap e_c}\left(CZ_{e_c\setminus q}^{(r)}\right)^{B_{e_c}}\right)\\
    =&\frac{1}{2^{{n \choose c}}}\left(\prod_{r=1}^{2\alpha}\prod_{i}\left(X_i^{(r)}\right)^{x_i^{(r)}}\cdot\sum_{\vec{B}}\prod_{r=1}^{2\alpha}\prod_{e_c}\prod_{q\subseteq v(\vec{x}^{(r)})\cap e_c}\left(CZ_{e_c\setminus q}^{(r)}\right)^{B_{e_c}}\right).
\end{aligned}
\end{equation}
Here in the final line, we put the summation of hypergraph configuration $\vec{B}$ into the phase-gate part, since the first Pauli-$X$ part is irrelevant to $\vec{B}$. This indicates that we can separate the Pauli-$X$ part and the phase-gate part.
 
As shown in Eq.~\eqref{eq:averho}, 
the average $\alpha$-fold density matrix $\overline{\rho^{(\alpha)}}$ is a linear combination of $\frac{1}{2^{{n \choose c}}}\sum_{\vec{B}}\prod_{r=1}^{2\alpha}\mathrm{St}(G_{n,c}^{(\vec{B})},\vec{x}^{(r)})$ indexed by $\{\vec{x}^{(r)}\}_{r=1}^{2\alpha}$. By focusing one specific term and inserting the result in Eq.~\eqref{eq:Halpha}, we get 
\begin{equation}\label{eq:canselX}
\begin{aligned}
    &\Tr{\frac{1}{2^{{n \choose c}}}\sum_{\vec{B}}\prod_{r=1}^{2\alpha}\mathrm{St}(G_{n,c}^{(\vec{B})},\vec{x}^{(r)})\cdot\bigotimes_i \Lambda_i^{(\alpha)}}\\
    =&\Tr\left\{\prod_{i}\left(\prod_{r=1}^{2\alpha}\left(X_i^{(r)}\right)^{x_i^{(r)}}\left(\id_i^{\otimes 2\alpha}+X_i^{\otimes 2\alpha}\right)\right)
    \prod_{i}\frac{\id_i^{\otimes 2\alpha}+Z_i^{\otimes 2\alpha}}{2}\cdot\frac{1}{2^{{n \choose c}}}\sum_{\vec{B}}\prod_{r=1}^{2\alpha}\prod_{e_c}\prod_{q\subseteq v(\vec{x}^{(r)})\cap e_c}\left(CZ_{e_c\setminus q}^{(r)}\right)^{B_{e_c}}\right\}\\
     =&\delta_{\vec{x},\vec{x}^{(1)},\cdots,\vec{x}^{(2\alpha)}}\Tr{C_0\cdot C_1(\vec{x})}.
\end{aligned}
\end{equation}
Here similar as before we separate the $X$-part and phase-gate part, and we change the order of the product in the $X$-part. Due to Lemma \ref{le:1}, the whole trace is nonzero iff $\prod_{r=1}^{2\alpha}\left(X_i^{(r)}\right)^{x_i^{(r)}}$ equals to $\id_i^{\otimes 2\alpha}$ or $X_i^{\otimes 2\alpha}$ on $2\alpha$-replica. It indicates that all the vectors $\vec{x}^{(1)},\cdots,\vec{x}^{(2\alpha)}$ should be identical, which is denoted by $\vec{x}$ in the final line in Eq.~\eqref{eq:canselX}. The physical meaning is that, to make nonzero contribution, the choice of stabilizers should be the same among $2\alpha$ replicas.
For short, we use $C_0$ to represent
\begin{equation}\label{eq:C0}
    C_0:=\prod_{i}\frac{\id_i^{\otimes 2\alpha}+Z_i^{\otimes 2\alpha}}{2}
\end{equation}
and $C_1(\vec{x})$ represents
\begin{equation}\label{eq:C1}
\begin{aligned}
    C_1(\vec{x}):=&\left.\frac{1}{2^{{n \choose c}}}\sum_{\vec{B}}\prod_{r=1}^{2\alpha}\prod_{e_c}\prod_{q\subseteq v(\vec{x}^{(r)})\cap e_c}\left(CZ_{e_c\setminus q}^{(r)}\right)^{B_{e_c}}\right|_{\vec{x}^{(r)}=\vec{x}}\\
    =&\frac{1}{2^{{n \choose c}}}\sum_{\vec{B}}\prod_{e_c}\prod_{q\subseteq v(\vec{x})\cap e_c}\left(\left(CZ_{e_c\setminus q}\right)^{\otimes 2\alpha}\right)^{B_{e_c}}\\
    =&\prod_{e_c}\frac{1}{2}\left(\id+\prod_{q\subseteq v(\vec{x})\cap e_c}\left(CZ_{e_c\setminus q}\right)^{\otimes 2\alpha}\right).
\end{aligned}
\end{equation}
In the final line, it is in fact an interchange of index: $\sum_{\vec{B}}\prod_{e_c}\rightarrow\prod_{e_c}\sum_{B_{e_c}}$.

Since $C_0$ and $C_1$ are both phase matrices in the $Z$ basis, the trace of their multiplication can be written as
\begin{equation}\label{eq:finalave}
\begin{aligned}
    \Tr{C_0\cdot C_1(\vec{x})}=&\sum_{\mathrm{T}}\bra{\mathrm{T}}\left(C_0\cdot C_1(\vec{x})\right)\ket{\mathrm{T}}\\
    =&\sum_{\mathrm{T}}\left(\bra{\mathrm{T}}C_0\ket{\mathrm{T}}\cdot\bra{\mathrm{T}}C_1(\vec{x})\ket{\mathrm{T}}\right)
\end{aligned}
\end{equation}
where $\mathrm{T}=(\vec{t}_1,\cdots,\vec{t_i},\cdots,\vec{t}_n)$ is a $2\alpha\times n$ binary matrix denoting the bit string in $Z$ basis whose elements are $t_{j,i}$. Each row of $\mathrm{T}$ represents one replica, and each column represents a corresponding qubit. The first term $\bra{\mathrm{T}}C_0\ket{\mathrm{T}}$ is either $0$ or $1$. It equals to $1$ iff $\bra{\mathrm{T}}\frac{\id_i^{\otimes 2\alpha}+Z_i^{\otimes 2\alpha}}{2}\ket{\mathrm{T}}=1$ for all $i$, i.e., $Z_i^{\otimes 2\alpha}\ket{\mathrm{T}}=\ket{\mathrm{T}}$ or equivalently
\begin{equation}\label{eq:con1}
    \sum_{j=1}^{2\alpha}t_{j,i}=0,\quad\forall i.
\end{equation}
On the other hand, the second term $\bra{\mathrm{T}}C_1(\vec{x})\ket{\mathrm{T}}=1$ iff 
\begin{equation}
\bra{\mathrm{T}}\frac{1}{2}\left(\id+\prod_{q\subseteq v(\vec{x})\cap e_c}\left(CZ_{e_c\setminus q}\right)^{\otimes 2\alpha}\right)\ket{\mathrm{T}}=1    
\end{equation}
for all $e_c$, i.e., $\prod_{q\subseteq v(\vec{x})\cap e_c}\left(CZ_{e_c\setminus q}\right)^{\otimes 2\alpha}\ket{\mathrm{T}}=\ket{\mathrm{T}}$ or equivalently
\begin{equation}\label{eq:con2}
    \sum_{q\subset e_c}\left(\prod_{v_i\in q}x_i\norm{\bigodot_{v_k\in e_c\setminus q}\vec{t_k}}_1\right)=0,\quad\forall |e_c|=c,
\end{equation}
where $\prod_{v_i\in q}x_i$ is by $\prod_{q\subseteq v(\vec{x})}$ and $\norm{\bigodot_{v_k\in e_c\setminus q}\vec{t_k}}_1$ is for the phase gate. Here we ignore the case $q=e_c$ since in this case $\left(CZ_{e_c\setminus q}\right)^{\otimes2\alpha}=(-1)^{2\alpha}=1$.
These two conditions are just that of Eq.~\eqref{eq:conhyper1} and Eq.~\eqref{eq:conhyper2} in main text.

Finally, plug Eq.~\eqref{eq:averho} and the final line of Eq.~\eqref{eq:canselX} into Eq.~\eqref{eq:avemagic} to get
\begin{equation}
\begin{aligned}
    \langle\mathbf{m}_{\alpha}\rangle_{\mc{E}_c}
    =&\Tr{\overline{\rho^{(\alpha)}}\cdot\bigotimes_i \Lambda_i^{(\alpha)}}\\
    =&\frac{1}{2^{2\alpha n}}\sum_{\vec{x}}\Tr{C_0\cdot C_1(\vec{x})}\\
    =&\frac{1}{2^{2\alpha n}}\sum_{\vec{x}}\sum_{\mathrm{T}}\left(\bra{\mathrm{T}}C_0\ket{\mathrm{T}}\cdot\bra{\mathrm{T}}C_1(\vec{x})\ket{\mathrm{T}}\right).
\end{aligned}
\end{equation}
It is clear that now the average is just related to the number of possible $\vec{x}$ and $\mathrm{T}$ satisfying conditions \eqref{eq:con1} and \eqref{eq:con2}, and we finish the proof.

\subsection{Proof of Theorem \ref{theo:randomvar}}\label{Apptheo:randomvar}
Similar to that of Theorem \ref{theo:randomave} in Sec.~\ref{App:Thmain-counting}, the first step is to write the moment of the PL-moment in the following replica form.
\begin{equation}
\begin{aligned}
    \mathbf{m}_{\alpha}(\ket{\Psi})^{\tau}&=\left(\Tr{\rho^{(\alpha)}(\Psi)\cdot\bigotimes_i \Lambda_i^{(\alpha)}}\right)^{\tau}\\
    &=\Tr{\left(\bigotimes_{s=1}^{\tau}\rho^{(\alpha)}_s(\Psi)\right)\cdot\bigotimes_{s=1}^{\tau}\left(\bigotimes_i \Lambda_{i,s}^{(\alpha)}\right)}\\
    &=\Tr{\rho^{(\tau\alpha)}(\Psi)\cdot\bigotimes_{s=1}^{\tau}\left(\bigotimes_i \Lambda_{i,s}^{(\alpha)}\right)}.
\end{aligned}
\end{equation}
In this way, we only need to replace $\alpha$ with $\tau\alpha$ in Eq.~\eqref{eq:averho}
and Eq.~\eqref{eq:Halpha}, and replace $\bigotimes_i \Lambda_{i,s}^{(\alpha)}$ with $\bigotimes_{s=1}^{\tau}\left(\bigotimes_i \Lambda_{i,s}^{(\alpha)}\right)$ in Eq.~\eqref{eq:canselX} to get (ignoring the $2^{-2\tau\alpha n}$ prefactor)
\begin{equation}
\begin{aligned}
    &\Tr{\frac{1}{2^{{n \choose c}}}\sum_{\vec{B}}\prod_{r=1}^{2\tau\alpha}\mathrm{St}(G_{n,c}^{(\vec{B})},\vec{x}^{(r)})\cdot\bigotimes_{s=1}^{\tau}\left(\bigotimes_i \Lambda_{i,s}^{(\alpha)}\right)}\\
    =&\prod_{s=1}^{\tau}\delta_{\vec{x}_s,\vec{x}^{(s\alpha+1)},\cdots,\vec{x}^{((s+1)\alpha)}}\cdot\Tr{\prod_{s=1}^{\tau}\left(\prod_{i}\frac{\id_{i,s}^{\otimes 2\alpha}+Z_{i,s}^{\otimes 2\alpha}}{2}\right)\cdot\sum_{\vec{B}} \frac{1}{2}\prod_{e_c} \prod_{s=1}^{\tau} \prod_{q\subseteq v(\vec{x}_s)\cap e_c}\left(\left(CZ^{(s)}_{e_c\setminus q}\right)^{\otimes 2\alpha}\right)^{B_{e_c}}}\\
    =&\prod_{s=1}^{\tau}\delta_{\vec{x}_s,\vec{x}^{(s\alpha+1)},\cdots,\vec{x}^{((s+1)\alpha)}}\cdot\Tr{\prod_{s=1}^{\tau}\left(\prod_{i}\frac{\id_{i,s}^{\otimes 2\alpha}+Z_{i,s}^{\otimes 2\alpha}}{2}\right)\cdot\prod_{e_c}\frac{1}{2}\left(\id+\prod_{s=1}^{\tau}\prod_{q\subseteq v(\vec{x}_s)\cap e_c}\left(CZ^{(s)}_{e_c\setminus q}\right)^{\otimes 2\alpha}\right)}\\
    =&\Tr{\prod_{s=1}^{\tau} C_{0,s}\cdot C_{1}(\mathrm{X})}
\end{aligned}
\end{equation}
where
\begin{equation}
    C_{0,s}:=\prod_{i}\frac{\id_{i,s}^{\otimes 2\alpha}+Z_{i,s}^{\otimes 2\alpha}}{2}
\end{equation}
and 
\begin{equation}
    C_{1}(\mathrm{X}):=\prod_{e_c}\frac{1}{2}\left(\id+\prod_{s=1}^{\tau}\prod_{q\subseteq v(\vec{x}_s)\cap e_c}\left(CZ^{(s)}_{e_c\setminus q}\right)^{\otimes 2\alpha}\right)
\end{equation}
where $\mathrm{X}=(\vec{x}_1,\cdots,\vec{x}_s,\cdots,\vec{x}_{\tau})$.

With a similar argument as that in Eq.~\eqref{eq:finalave}, one can transform these phase matrices to the conditions Eq.~\eqref{eq:convar1} and Eq.~\eqref{eq:convar2} according to the expression of $C_{0,s}$ and $C_{1}(\mathrm{X})$ respectively.

\subsection{Proof of Proposition \ref{prop:general}}\label{ApppropRecursion}
We prove the result mainly by recursion. First define a series of more general conditions $C_{\xi}$ with parameter $\xi\in[c]$ describing the cardinality of the hyperedge.
\begin{equation}\label{eq:ksi}
    C_{\xi}:\quad \sum_{q\subset e_{\xi}}\left(\prod_{v_i\in q}x_i\norm{\bigodot_{v_k\in e_{\xi}\setminus q}\vec{t_k}}_1\right)=0,\quad\forall |e_{\xi}|=\xi.
\end{equation}
Notice that as in main text here $q\neq\emptyset$ unless otherwise specified. It is clear the condition Eq.~\eqref{eq:conhyper2} is just $C_{\xi=c}$. 

We further add the case $q=\emptyset$ and define another similar series of conditions $C_{\xi}^{*}$ for $\xi\in[c]$ as
\begin{equation}\label{eq:ksistar}
    C_{\xi}^{*}:\sum_{q\subset e_{\xi}\| q=\emptyset}\left(\prod_{v_i\in q}x_i\norm{\bigodot_{v_k\in e_{\xi}\setminus q}\vec{t_k}}_1\right):= \sum_{q\subset e_{\xi}}\left(\prod_{v_i\in q}x_i\norm{\bigodot_{v_k\in e_{\xi}\setminus q}\vec{t_k}}_1\right)+\norm{\bigodot_{v_k\in e_{\xi}}\vec{t_k}}_1=0,\forall |e_{\xi}|=\xi.
\end{equation}
For convenience, we also use $C_{\xi}$ (or $C_{\xi}^{*}$) to denote all the 2-tuples $(\mathrm{T},\vec{x})$ satisfying $C_{\xi}$ (or $C_{\xi}^{*}$) and the parity constraint in Eq.~\eqref{eq:conhyper1} in main text. Denote $N_{\xi}=|C_{\xi}|$ (and $N_{\xi}^{*}=|C_{\xi}^{*}|$) as the number of them, so $\mathrm{N}(c,\alpha,n)=N_c$ in Theorem \ref{theo:randomave}. 

Furthermore, define $\Delta C_{\xi}^{*}$ as the set of all the $(\mathrm{T},\vec{x})$ satisfying condition $C_{\xi}^{*}$ but not satisfying condition $C_{\xi-1}^{*}$ and denote $\Delta N_{\xi}^{*}=|\Delta C_{\xi}^{*}|$. Similarly, define $\Delta C_{\xi}$ as the set of all the $(\mathrm{T},\vec{x})$ satisfying condition $C_{\xi}$ but not satisfying condition $C_{\xi-1}^{*}$ and denote $\Delta N_{\xi}=|\Delta C_{\xi}|$. Notice that in the definition of $\Delta C_{\xi}$, the set is $C_{\xi-1}^{*}$ but not $C_{\xi-1}$. Then we directly have the inequalities of the set cardinality $N_{\xi}^{*}\leq N_{\xi-1}^{*}+\Delta N_{\xi}^{*}$ and $N_{\xi}\leq N_{\xi-1}^{*}+\Delta N_{\xi}$, and we give the estimation of both $\Delta N_{\xi}^{*}$ and $\Delta N_{\xi}$ as follows.

By definition, for any element $(\mathrm{T},\vec{x})$ in $\Delta C_{\xi}^{*}$, one can find a set $e_{\xi-1}^{(0)}$ such that 
\begin{equation}\label{eq:generalboundcon}
    \sum_{q\subset e_{\xi-1}^{(0)}\| q=\emptyset}\left(\prod_{v_i\in q}x_i\norm{\bigodot_{v_k\in e_{\xi-1}^{(0)}\setminus q}\vec{t_k}}_1\right)=1.
\end{equation}
Eq.~\eqref{eq:generalboundcon} results in the following powerful restriction on the elements in $C_{\xi}^{*}$. For any $v_{i_1},v_{i_2}\notin e_{\xi-1}^{(0)}$, suppose their corresponding $\vec{t}$ are the same, i.e., $\vec{t}_{i_1}=\vec{t}_{i_2}$, then consider two sets $e_{\xi}^{(1)}=\{v_{i_1}\}\cup e_{\xi-1}^{(0)}$ and $e_{\xi}^{(2)}=\{v_{i_2}\}\cup e_{\xi-1}^{(0)}$. 
\begin{equation}\label{eq:generalbound1}
\begin{aligned}
    0=&\sum_{q\subset e_{\xi}^{(1)}\|q=\emptyset}\left(\prod_{v_i\in q}x_i\norm{\bigodot_{v_k\in e_{\xi}^{(1)}\setminus q}\vec{t_k}}_1\right)-\sum_{q\subset e_{\xi}^{(2)}\|q=\emptyset}\left(\prod_{v_i\in q}x_i\norm{\bigodot_{v_k\in e_{\xi}^{(2)}\setminus q}\vec{t_k}}_1\right)\\
    =&\sum_{q\subset e_{\xi}^{(1)}}\left(\prod_{v_i\in q}x_i\norm{\bigodot_{v_k\in e_{\xi}^{(1)}\setminus q}\vec{t_k}}_1\right)-\sum_{q\subset e_{\xi}^{(2)}}\left(\prod_{v_i\in q}x_i\norm{\bigodot_{v_k\in e_{\xi}^{(2)}\setminus q}\vec{t_k}}_1\right)\\
    =&\left(x_{i_1}-x_{i_2}\right)\sum_{q\subset e_{\xi-1}^{(0)}\|q=\emptyset}\left(\prod_{v_i\in q}x_i\norm{\bigodot_{v_k\in e_{\xi-1}^{(0)}\setminus q}\vec{t_k}}_1\right)\\
    =&x_{i_1}-x_{i_2}.
\end{aligned}
\end{equation}
In the first line, we use the precondition that $C_{\xi}^{*}$ is satisfied.
In the third line, we use the fact that $\vec{t}_{i_1}=\vec{t}_{i_2}$ so those two summations in differ only in $x_{i_1}$ and $x_{i_2}$. The final line is by Eq.~\eqref{eq:generalboundcon}.
As a result, Eq.~\eqref{eq:generalbound1} shows that given $\mathrm{T}$, all the vertices except those in $e_{\xi-1}^{(0)}$ share the same valid vector $\vec{t}$, should have the same value of $x_i$.

Similarly, for any element $(\mathrm{T},\vec{x})$ in $\Delta C_{\xi}$, one can use the same strategy to find the same result starting from the second line of Eq.~\eqref{eq:generalbound1}. That is, given $\mathrm{T}$, all the vertices except those in $e_{\xi-1}^{(0)}$ which share the same valid vector $\vec{t}$, should take the same value of $x_i$.
 
Now fix $\mathrm{T}$, consider $\vec{x}$ such that $(\mathrm{T},\vec{x})\in \Delta C_{\xi}^{*}$. There are totally $2^{2\alpha-1}$ possible $\vec{t}$ satisfying Eq.~\eqref{eq:conhyper1}, the number of choices of $\vec{x}$ is no more than $2^{\xi-1}\cdot 2^{2^{2\alpha-1}}$. Here $2^{\xi-1}$ denotes the number of choices of $x_i$ whose $v_i\in e_{\xi-1}^{(0)}$ in Eq.~\eqref{eq:generalboundcon}. $2^{2^{2\alpha-1}}$ denotes the number of choices of other $x_i$ outside $e_{\xi-1}^{(0)}$, which should be the same for the same-valued $\vec{t}_i$. Therefore, by counting all $2^{(2\alpha-1)n}$ possible $\mathrm{T}$, the size of $\Delta C_{\xi}^{*}$ is
\begin{equation}
    \Delta N_{\xi}^{*}\leq 2^{(2\alpha-1)n}\cdot 2^{\xi-1}\cdot 2^{2^{2\alpha-1}}.
\end{equation}
One can similarly get 
\begin{equation}
    \Delta N_{\xi}\leq 2^{(2\alpha-1)n}\cdot 2^{\xi-1}\cdot 2^{2^{2\alpha-1}}.
\end{equation}

For the initial condition when $\xi=2$, $C_{2}^{*}$ becomes
\begin{equation}\label{eq:C2star}
\begin{aligned}
    C_{2}^{*}:\quad \sum_{q\subset e_{2}\|q=\emptyset}\left(\prod_{v_i\in q}x_i\norm{\bigodot_{v_k\in e_{2}\setminus q}\vec{t_k}}_1\right)&=x_{i_1}\norm{\vec{t}_{i_2}}_1+x_{i_2}\norm{\vec{t}_{i_1}}_1+\norm{\vec{t}_{i_1}\odot\vec{t}_{i_2}}_1\\
    &=\norm{\vec{t}_{i_1}\odot\vec{t}_{i_2}}_1=0,\quad\forall i_1,i_2
\end{aligned}   
\end{equation}
where we use the fact that $\norm{\vec{t}_{i}}_1=0$.
Although it gives no limitation on $\vec{x}$, it limits the valid $\mathrm{T}$. Indeed, this is a matching problem, so there are no more than $(2\alpha-1)!!\cdot (2^{\alpha})^n$ possible $\mathrm{T}$, i.e.,
\begin{equation}\label{eq:C2starC}
    N_2^{*}\leq (2\alpha-1)!!\cdot (2^{\alpha})^n\cdot 2^n=(2\alpha-1)!!\cdot (2^{\alpha+1})^n.
\end{equation}
Then by recursion, one has
\begin{equation}
\begin{aligned}
   \mathrm{N}(c,\alpha,n)=&N_c\leq N_{c-1}^{*}+\Delta N_c\leq N^*_2+\sum_{\xi=3}^{c-1}\Delta N_{\xi}^{*}+\Delta N_c\\
    \leq& (2\alpha-1)!!\cdot (2^{\alpha+1})^n+(2^{c}-4)\cdot 2^{2^{2\alpha-1}}\cdot 2^{(2\alpha-1)n}\leq 2^{c}\cdot 2^{2^{2\alpha-1}}\cdot 2^{(2\alpha-1)n}
\end{aligned}
\end{equation}
and 
\begin{equation}
    \langle\mathbf{m}_{\alpha}\rangle_{\mc{E}_c}=\frac{\mathrm{N}(c,\alpha,n)}{2^{2\alpha n}}\leq \frac{2^{(c+2^{2\alpha-1})}}{2^n}
\end{equation}
for $\alpha\geq 2$.

We remark that in the final step, we use $(2\alpha-1)!!\leq 4\cdot 2^{2^{2\alpha-1}}$ as $\alpha\geq 2$ for simplicity. But such inequality is quite loose, so one can use the original value to give a relatively tighter bound.

\section{Proofs of the concentration of magic for random 3-uniform hypergraph states}
All the proofs in this section are based on the simplified constraint shown in Corollary \ref{co:3special}.
\subsection{Proof of Proposition \ref{prop:3cardi}}\label{Appprop2}
When $c=3$, we inherit the spirit of 
the proof of Proposition \ref{prop:general} in Sec~\ref{ApppropRecursion} by considering $C_2^*$ in Eq.~\eqref{eq:C2star} and $C_3$ in Corollary \ref{co:3special}. Suppose one finds two vectors satisfying $\norm{\vec{t}_{i_1}\odot\vec{t}_{i_2}}_1=1$, then besides $i_1,i_2$, the $x_i$ whose corresponding $\vec{t_i}$ share the same-valued of vector, should also be the same.

In fact, we go one step further to show that these result also applies to $i_1,i_2$ as follows. To be specific, suppose one finds two vectors satisfying $\norm{\vec{t}_{i_1}\odot\vec{t}_{i_2}}_1=1$, and also $\vec{t}_{i_3}=\vec{t}_{i_1}$, then focus on the edge $e_3=\{v_{i_1},v_{i_2},v_{i_3}\}$, Eq.~\eqref{eq:conhyper3edge} shows
\begin{equation}
    x_{i_1}\norm{\vec{t}_{i_2}\odot \vec{t}_{i_3}}_1+x_{i_2}\norm{\vec{t}_{i_1}\odot \vec{t}_{i_3}}_1+x_{i_3}\norm{\vec{t}_{i_1}\odot \vec{t}_{i_2}}_1=(x_{i_1}+x_{i_3})\norm{\vec{t}_{i_1}\odot \vec{t}_{i_2}}_1=x_{i_1}+x_{i_3}=0,
\end{equation}
since $\norm{\vec{t}_{i_1}\odot \vec{t}_{i_3}}_1=\norm{\vec{t}_{i_1}}_1=0$. 
On the other hand, if there is no such pair with  $\norm{\vec{t}_{i_1}\odot\vec{t}_{i_2}}_1=1$,  there is no constraint on $\vec{x}$ at all.

We define the set $\mc{D}=\{\vec{d}~\big|~\|\vec{d}\|_1=0\}$ with totally $2^{2\alpha-1}$ elements. Actually, it is easy to check $\vec{d}$ is equivalent to $\vec{1}-\vec{d}$ for the $c=3$ constraint in Eq.~\eqref{eq:conhyper3edge}. That is, if a 2-tuple $2$-tuple 
$(\mathrm{T},\vec{x})$ fulfills the constraint, then one can substitute any $\vec{t}$ to $\vec{1}-\vec{t}$. As a result, there are $2^{2\alpha-2}$ nonequivalent $\vec{d}$.

Focusing on the $\alpha=2$ case, there are $4$ nonequivalent $\vec{d}$ showing: $\vec{d}_0=(0,0,0,0)$, $\vec{d}_1=(1,1,0,0)$, $\vec{d}_2=(1,0,1,0)$, $\vec{d}_3=(1,0,0,1)$, and their relations are 
\begin{equation}\label{eq:app:d}
\begin{aligned}
&\norm{\vec{d}_0\odot\vec{d}_1}_1=\norm{\vec{d}_0\odot\vec{d}_2}_1=\norm{\vec{d}_0\odot\vec{d}_3}_1=0,\\
&\norm{\vec{d}_1\odot\vec{d}_2}_1=\norm{\vec{d}_1\odot\vec{d}_3}_1=\norm{\vec{d}_2\odot\vec{d}_3}_1=1.
\end{aligned}
\end{equation}
Use the function $n(\vec{d})$ to denote the number of $\vec{t}$ that equals to $\vec{d}$. According to the previous analysis, if at most one of $n(\vec{d}_1),n(\vec{d}_2),n(\vec{d}_3)$ is nonzero, $\vec{x}$ can take all the $2^n$ values. On the other hand, there exists vector-pair whose Hadamard product is $1$, and the $x_i$ should be the same if the corresponding $t_i$ are the same. We thus use $x_{\vec{d}_0}$, $x_{\vec{d}_1}$, $x_{\vec{d}_2}$, $x_{\vec{d}_3}$ to denote the possible values, if their corresponding $n(\vec{d})$ is nonzero. Actually, we find that there are only $4$ possible $\vec{x}$ in any of these cases. First, $x_{\vec{d}_0}$ must keep zero, so the existence of $\vec{d}_0$ in $\mathrm{T}$ does not affect the possibility of $\vec{x}$. If only two of $n(\vec{d}_1),n(\vec{d}_2),n(\vec{d}_3)$ are nonzero, one can check that $x_{\vec{d}_{i\neq 0}}$ can take any value and thus there are $4$ possibilities; if all three are nonzero, it gives the constraint $x_{\vec{d}_1}+x_{\vec{d}_2}+x_{\vec{d}_3}=1$, and thus there are still $4$ possibilities.

As a result, 
\begin{equation}
\begin{aligned}
    \langle\mathbf{m}_{2}\rangle_{\mc{E}_3}&=\frac{1}{2^{4n}}\cdot2^n\cdot\left(2^n(3\cdot 2^n-2)+4(4^n-3\cdot 2^n+2)\right)\\
    &=\frac{7}{2^n}-\frac{14}{4^n}+\frac{8}{8^n}
\end{aligned}
\end{equation}
where the first $2^n$ comes from the equivalence between $\vec{d}$ and $\vec{1}-\vec{d}$ and $3\cdot 2^n-2$ is the number of cases that at least two of $n(\vec{d}_1),n(\vec{d}_2),n(\vec{d}_3)$ equal to $0$.

When $\alpha\geq3$, there are totally $2^{2\alpha-2}$ nonequivalent $\vec{d}$. On the one hand, if all $\norm{\vec{t}_{i_1}\odot\vec{t}_{i_2}}_1=0$ with $\vec{t}_{i_1}, \vec{t}_{i_2}\in \mathrm{T}  $, there is no constraint on $\vec{x}$, which is actually $C_2^*$ in Eq.~\eqref{eq:C2star} with couting result in Eq.~\eqref{eq:C2starC}.  
On the one hand, if $\norm{\vec{d}\odot\vec{d}'}_1=1$ and both $n(\vec{d})\neq0$, $n(\vec{d}')\neq0$, it is not hard to see that as long as $x_{\vec{d}}$ and $x_{\vec{d}'}$ are fixed, all the other $x_i$ are also determined. To be specific, considering $e_3=\{v_{i_1},v_{i_2},v_{i_3}\}$ with $\vec{t}_{i_1}=\vec{d}$, $\vec{t}_{i_2}=\vec{d}'$, $x_{i_3}$ is determined by the constraint in Eq.~\eqref{eq:conhyper3edge}. As a result, there are $4$ possibilities when choosing $x_{\vec{d}}$ and $x_{\vec{d}'}$, and the number of possible $\mathrm{T}$ should be no more than the maximal value $2^{(2\alpha-1)n}$.

We finally have the upper bound 
\begin{equation}
\begin{aligned}
    \langle\mathbf{m}_{\alpha}\rangle_{\mc{E}_3}&\leq\frac{1}{2^{2\alpha n}}\cdot\left(2^n\cdot(2\alpha-1)!!(2^{\alpha})^n+4\cdot2^{(2\alpha-1)n}\right)\\
    &=\frac{4}{2^n}+\frac{(2\alpha-1)!!}{2^{(\alpha-1)n}}.
\end{aligned}
\end{equation}

\subsection{Proof of Proposition \ref{prop:var} }\label{Appprop:var} 
From Theorem \ref{theo:randomvar}, one can directly derive the corresponding counting problem when $\alpha=2$ and $c=3$:
\begin{equation}
\begin{gathered}
    \norm{\vec{t}^{(1)}_i}_1=\norm{\vec{t}^{(2)}_i}_1=0,\quad\forall i;\\
    \sum_{i\in e_3}\left(x_{i,1}\norm{\bigodot_{k\neq i}\vec{t}^{(1)}_k}_1+x_{i,2}\norm{\bigodot_{k\neq i}\vec{t}^{(2)}_k}_1\right)=0,\quad\forall |e_3|=3.
\end{gathered}
\end{equation}

Following the proof of Proposition \ref{prop:general} in Sec.~\ref{ApppropRecursion} and Proposition \ref{prop:3cardi} in Sec.~\ref{Appprop2}, similar conditions can be derived in the same way.
\begin{itemize}
    \item[a)] Suppose we can find $v_{k_1}$ and $v_{k_2}$ such that $\norm{\vec{t}^{(1)}_{k_1}\odot\vec{t}^{(1)}_{k_2}}_1=1$ and $\norm{\vec{t}^{(2)}_{k_1}\odot\vec{t}^{(2)}_{k_2}}_1=0$. Then for any other two $v_{i_1}$ and $v_{i_2}$ with the same $(\vec{t}^{(1)},\vec{t}^{(2)})$, consider 3-hyperedges $e_{3,1}=\{v_{i_1},v_{k_1},v_{k_2}\}$ and $e_{3,2}=\{v_{i_2},v_{k_1},v_{k_2}\}$, 
    \begin{equation}
    \begin{aligned}
        0=&\sum_{i\in e_{3,1}}\left(x_{i,1}\norm{\bigodot_{k\neq i}\vec{t}^{(1)}_k}_1+x_{i,2}\norm{\bigodot_{k\neq i}\vec{t}^{(2)}_k}_1\right)-\sum_{i\in e_{3,2}}\left(x_{i,1}\norm{\bigodot_{k\neq i}\vec{t}^{(1)}_k}_1+x_{i,2}\norm{\bigodot_{k\neq i}\vec{t}^{(2)}_k}_1\right)\\
        =&(x_{i_1,1}-x_{i_2,1})\norm{\bigodot_{k\neq i}\vec{t}^{(1)}_k}_1+(x_{i_1,2}-x_{i_2,2})\norm{\bigodot_{k\neq i}\vec{t}^{(2)}_k}_1\\
        =&x_{i_1,1}-x_{i_2,1}.
    \end{aligned}
    \end{equation}
    This indicates that $x_{i,1}$ corresponding to the same $(\vec{t}^{(1)}_i,\vec{t}^{(2)}_i)$ should be the same except for $v_{k_1}$ and $v_{k_2}$. With a similar analysis as that in Appendix \ref{Appprop2}, one can remove the exception by considering $e_3=\{v_{k_1},v_{k_2},v_{k_3}\}$ where $(\vec{t}^{(1)}_{k_3},\vec{t}^{(2)}_{k_3})=(\vec{t}^{(1)}_{k_1},\vec{t}^{(2)}_{k_1})$.
    \item[b)] Suppose we can find $v_{k_1}$ and $v_{k_2}$ such that $\norm{\vec{t}^{(1)}_{k_1}\odot\vec{t}^{(1)}_{k_2}}_1=0$ and $\norm{\vec{t}^{(2)}_{k_1}\odot\vec{t}^{(2)}_{k_2}}_1=1$. One can derive a similar conclusion: $x_{i,2}$ corresponding to the same $(\vec{t}^{(1)}_i,\vec{t}^{(2)}_i)$ should be the same.
    \item[c)] Suppose we can find $v_{k_1}$ and $v_{k_2}$ such that $\norm{\vec{t}^{(1)}_{k_1}\odot\vec{t}^{(1)}_{k_2}}_1=1$ and $\norm{\vec{t}^{(2)}_{k_1}\odot\vec{t}^{(2)}_{k_2}}_1=1$. One can derive a similar conclusion: $x_{i,1}+x_{i,2}$ corresponding to the same $(\vec{t}^{(1)}_i,\vec{t}^{(2)}_i)$ should be the same.
\end{itemize}
Any two of the conditions in cases a), b) and c) could lead to the conclusion that the $(x_{i,1},x_{i_2})$ whose corresponding $\vec{t_i}^{(1)},\vec{t_i}^{(2)}$ share the same-valued of vector, should also be the same. In this way, the number of choices of $\mathrm{X}$ is no more than $4^2$ according to the analysis in Sec.~\ref{Appprop2}. Also, if only one of the conditions in cases a), b) and c) is satisfied, the number of choices of $\mathrm{X}$ should be no more than $4\cdot 2^n$. Otherwise, one should count all the $4^n$ different $\mathrm{X}$. 


Since $\alpha=2$ here, there are still only $4$ nonequivalent $\vec{d}$ as in the proof of Proposition \ref{prop:3cardi}: $\vec{d}_0$, $\vec{d}_1$, $\vec{d}_2$, $\vec{d}_3$. According to Sec.~\ref{Appprop2}, when none of the assumptions in case a), b) and c) is satisfied, we should have $\norm{\vec{t}^{(1)}_{k_1}\odot\vec{t}^{(1)}_{k_2}}_1\equiv 0$ and $\norm{\vec{t}^{(2)}_{k_1}\odot\vec{t}^{(2)}_{k_2}}_1\equiv 0$ so the number of possible $\mathcal{T}$ is $(3\cdot 2^n-2)^2$. When only the assumption in case a) is satisfied, we should have $\norm{\vec{t}^{(2)}_{k_1}\odot\vec{t}^{(2)}_{k_2}}_1\equiv 0$ so the number of choices of $\mathrm{T}^{(2)}$ is $3\cdot 2^n-2$. In this case, since the number of choices of $\mathrm{T}^{(1)}$ is no more than $4^n-3\cdot 2^n+2$, the total number of $(\mathcal{T},\mathrm{X})$ is upper bounded by $4\cdot 2^n\cdot (4^n-3\cdot 2^n+2)\cdot(3\cdot 2^n-2)$. The same result holds for that in case b). When only the assumption in case c) is satisfied, the total number of $\mathcal{T}$ is no more than $15\cdot 4^n$, where $15$ comes from all possible choices of $4$ different $(\vec{d}^{(1)},\vec{d}^{(2)})$ such that only c) is satisfied. When any two of the conditions in cases a), b) and c) are satisfied, the number of possible $\mathcal{T}$ is no more than $(4^n-3\cdot 2^n+2)^2$. With the above conditions, by considering the equivalence between $\vec{d}$ and $\vec{1}-\vec{d}$ one can obtain in total
\begin{equation}
\begin{aligned}
    \langle\mathbf{m}_{2}^2\rangle_{\mc{E}_3} & \leq\frac{2^{2n}}{2^{8n}}(4^2\cdot (4^n-3\cdot 2^n+2)^2+2\cdot 4\cdot 2^n\cdot (4^n-3\cdot 2^n+2)\cdot(3\cdot 2^n-2)\\
    & \quad +4^n\cdot(3\cdot 2^n-2)^2+4\cdot 2^n\cdot 15\cdot 4^n)\\
    &=\frac{2^{2n}}{2^{8n}}\left(\left(2^n(3\cdot 2^n-2)+4(4^n-3\cdot 2^n+2)\right)^2+60\cdot 2^{3n}\right).
\end{aligned}
\end{equation}
Then one gets
\begin{equation}
    \delta^2_{\mc{E}_3}[\mathbf{m}_{2}]\leq\frac{60}{2^{3n}}.
\end{equation}

\subsection{Proof of Corollary \ref{co:concentarte}}\label{Appco:var}
\begin{proof}
\begin{equation}
\begin{aligned}
    &\Pr{\mb{M}_2(\ket{G_{n,3}})\leq n-3}\\
     =&\Pr{\mb{m}_2(\ket{G_{n,3}})\geq\frac{8}{2^n}}\\
    \leq&\Pr{\left|\mathbf{m}_2(\ket{G_{n,3}})-\langle\mathbf{m}_{2}\rangle_{\mc{E}_3}\right|\geq2^{-n}}\\
   \leq&\frac{\delta_{\mc{E}_3}[\mathbf{m}_{2}]^2}{\left(2^{-n}\right)^2}\leq\frac{60}{2^{n}}.
\end{aligned}
\end{equation}
with the final line by Chebyshev's inequality. 
 \end{proof}

\section{Proof of Theorem \ref{theo:p-ave}: average magic for non-uniform ensembles}\label{Apptheo:pave}
We can directly follow through the proof of Theorem \ref{theo:randomave} to Eq.~\eqref{eq:canselX}, Eq.~\eqref{eq:C0} and Eq.~\eqref{eq:C1}. The only difference is to replace the two coefficients $\frac{1}{2}$ with $(1-p)$ and $p$ in the expression of $C_1(\vec{x})$ to get
\begin{equation}
\begin{gathered}
    C_0=\prod_{i}\frac{\id_i^{\otimes 2\alpha}+Z_i^{\otimes 4}}{2},\\
    C_1(\vec{x},p)=\Tr{\prod_{e_3}\left((1-p)\id+p\prod_{v_i\in e_3}\left(\left(CZ_{e_3\setminus \{v_i\}}\right)^{\otimes 4}\right)^{x_i}\right)}.
\end{gathered}
\end{equation}
The expectation value of the PL-moment is still the summation
\begin{equation}\label{eq:sump}
\begin{aligned}
    \langle\mathbf{m}_{2}\rangle_{\mc{E}_3^p}
    =&\frac{1}{2^{4n}}\sum_{\vec{x}}\Tr{C_0\cdot C_1(\vec{x},p)}\\
    =&\frac{1}{2^{4n}}\sum_{\vec{x}}\sum_{\mathrm{T}}\left(\bra{\mathrm{T}}C_0\ket{\mathrm{T}}\cdot\bra{\mathrm{T}}C_1(\vec{x},p)\ket{\mathrm{T}}\right).
\end{aligned}
\end{equation}
First, to get nonzero $\Tr{C_0\cdot C_1(\vec{x},p)}$, all the $\mathrm{T}$ should still satisfy 
$\norm{\vec{t_i}}_1=0, \forall i$, i.e., $\bra{\mathrm{T}}C_0\ket{\mathrm{T}}=1$. For a given $\mathrm{T}$, the second term shows 
\begin{equation}\label{eq:app:e3}
    \bra{\mathrm{T}}C_1(\vec{x},p)\ket{\mathrm{T}}=\prod_{e_3}\left(1-p+p\bra{\mathrm{T}}\prod_{i\in e_3}\left(\left(CZ_{e_3\setminus \{i\}}\right)^{\otimes 4}\right)^{x_i}\ket{\mathrm{T}}\right).
\end{equation}
For each edge $e_3$, the corresponding term contributes a multiplication $(1-2p)$ if \begin{equation}
    \bra{\mathrm{T}}\prod_{i\in e_3}\left(\left(CZ_{e_3\setminus \{i\}}\right)^{\otimes 4}\right)^{x_i}\ket{\mathrm{T}}=-1.
\end{equation}
Otherwise, the multiplication is $1$.  The problem is thus transformed to counting the number of $e_3$ such that $\bra{\mathrm{T}}\prod_{i\in e_3}\left(\left(CZ_{e_3\setminus \{i\}}\right)^{\otimes 4}\right)^{x_i}\ket{\mathrm{T}}=-1$, i.e.,
\begin{equation}
    \sum_{v_i\in e_3}\left(x_i\norm{\bigodot_{k\neq i}\vec{t_k}}_1\right)=1
\end{equation}
given $\mathrm{T}$ and $\vec{x}$. 

$\vec{t_i}$ are binary vectors of dimension $4$, with the equivalence consideration that $\vec{t}\Leftrightarrow \vec{1}-\vec{t}$, there are $4$ nonequivalent vectors $\vec{d}_0=(0,0,0,0)$, $\vec{d}_1=(1,1,0,0)$, $\vec{d}_2=(1,0,1,0)$, $\vec{d}_3=(1,0,0,1)$. We adopt $\kappa_k^{\pm}$ to denote the number of $i$ such that $\vec{t_i}=\vec{d}_k$ and $x_i=1/0$. 
Since we should run over all $e_3$ in Eq.~\eqref{eq:app:e3}, the number of such $e_3$ is only related to $\kappa_k^{\pm}$, denoted by $f(\vec{\kappa})$.

With the relation of $\vec{d}_0,\vec{d}_1,\vec{d}_2,\vec{d}_3$ as shown in Eq.~\eqref{eq:app:d}, one can count the number of such $e_3$ and get $f(\vec{\kappa})$ in Eq.~\eqref{eq:fa} as the number of $(1-2p)$ in $\bra{\mathrm{T}}C_1(\vec{x},p)\ket{\mathrm{T}}$.
\begin{equation}
\begin{aligned}
    \langle\mathbf{m}_{2}\rangle_{\mc{E}_3^p}&=\frac{1}{2^{4 n}}\sum_{\vec{x}}\sum_{\mathrm{T}} (1-2p)^{f(\vec{\kappa})}\\
    &=\frac{1}{2^{4n}}\sum_{\vec{\kappa}}{n \choose \vec{\kappa}}\cdot 2^n(1-2p)^{f(\vec{\kappa})}
\end{aligned}
\end{equation}
where the term $2^n$ comes from the $\vec{t}\Leftrightarrow \vec{1}-\vec{t}$ equivalence.

\section{Proofs of the magic of hypergraph states with permutation symmetry}\label{app:sym}
\subsection{Proof of Proposition \ref{th:3-com}}
Our proof is mainly based on the magic formula in Eq.~\eqref{eq:3magic}.
With the permutation symmetry, the structure of the induced hypergraph $G^*_{\vec{x},\vec{z}}$ is only determined by $m=|A_x|$, $m_1=|A_{z,x}|$ and $m_0=|A_{z,-x}|$. There are two cases determined by the parity of $m$, as shown in Fig.~\ref{fig:indGsym} (a).
\begin{itemize}
    \item[a)] $m$ is odd. In this case, if both $v_j,v_k\in A_x$ or $v_j,v_k\in A_{-x}$, the number of $v_i$ such that $x_i=1$ is $m-2$ or $m$, then $b_{j,k}(\vec{x})=\sum_{\{v_i,v_j,v_k\}\in E}x_i=1$. If $v_j\in A_x,\ v_k\in A_{-x}$ or $v_k\in A_x,\ v_j\in A_{-x}$, the number of $v_i$ such that $x_i=1$ is $m-1$, then $b_{j,k}(\vec{x})=\sum_{\{v_i,v_j,v_k\}\in E}x_i=0$. Thus, the edge set $E^{(2)}_{\vec{x}}$ represented by function $b_{j,k}(\vec{x})$ describes two complete graphs of $m$ nodes and $n-m$ nodes, which means that 
    \begin{equation}\label{eq:sepUG}
        \Tr{U(G_{\vec{x},\vec{z}})}=\Tr{U(G^{\text{com}}_{m,m_1})}\Tr{U(G^{\text{com}}_{n-m,m_0})}
    \end{equation}
    where $G^{\text{com}}_{m,m'}$ represents a graph with $m$ vertices, complete $2$-edges and $m'$ $1$-edges. Thus one can calculate the trace with the summation of Boolean functions:
    \begin{equation}
    \begin{aligned}
        \Tr{U(G^{\text{com}}_{m,m'})}=&\sum_{\vec{a}}(-1)^{\sum_{i=1}^{m'}a_i+\sum_{j\neq k}a_j a_k}\\
        =&\sum_{w}\sum_{w'}{{m-m'} \choose w}{m' \choose w'}(-1)^{w'+\frac{1}{2}(w+w')(w+w'-1)}
    \end{aligned}
    \end{equation}
    
    where $w'=|\{i\leq m'|a_i=1\}|$ and $w=|\{i>m'|a_i=1\}|$ represents the number of $a_i=1$ with $z_i=1$ and $z_i=0$ respectively. Notice that we are concerned about not the individual Pauli-Liouville components but the summation of their fourth power and the summation of their absolute value. With the separability in Eq.~\eqref{eq:sepUG}, each term shows
    \begin{equation}
    \begin{aligned}
	\sum_{\vec{z}}\left(\Tr{U(G_{\vec{x},\vec{z}})}\right)^4
        =&\sum_{\vec{z}}\left(\Tr{U(G^{\text{com}}_{m,m_1})}\Tr{U(G^{\text{com}}_{n-m,m_0})}\right)^4\\
	=&\sum_{m_1=0}^m{m \choose m_1}\left(\Tr{U(G^{\text{com}}_{m,m_1})}\right)^4\cdot\sum_{m_0=0}^{n-m}{n-m \choose m_0}\left(\Tr{U(G^{\text{com}}_{n-m,m_0})}\right)^4,
    \end{aligned}
    \end{equation}
    and similar for that of absolute value. After some tedious calculation, one can derive
    \begin{equation}
    \begin{gathered}
	\sum_{m'=0}^m{m \choose m'}\left(\Tr{U(G^{\text{com}}_{m,m'})}\right)^4=\frac{2^{4m}}{2^{m+\frac{(-1)^m-1}{2}}},\\
	\sum_{m'=0}^m{m \choose m'}\left|\Tr{U(G^{\text{com}}_{m,m_1})}\right|=2^{\frac{3m}{2}+\frac{(-1)^m-1}{4}}.
    \end{gathered}
    \end{equation}
    Since $m$ is odd, one can obtain
    \begin{equation}
    \begin{aligned}
	\sum_{\vec{z}}\left(\Tr{P_{\vec{x},\vec{z}}\kb{G}}\right)^4=&\frac{1}{2^{4n}}\cdot\frac{2^{4m}}{2^{m+\frac{(-1)^m-1}{2}}}\cdot\frac{2^{4(n-m)}}{2^{n-m+\frac{(-1)^{n-m}-1}{2}}}=\frac{1}{2^{n-\frac{3+(-1)^n}{2}}},
    \end{aligned}
    \end{equation}
    and 
    \begin{equation}
    \begin{aligned}
	\sum_{\vec{z}}\left|\Tr{P_{\vec{x},\vec{z}}\kb{G}}\right|=&\frac{1}{2^{n}}\cdot2^{\frac{3m}{2}+\frac{(-1)^m-1}{4}}\cdot2^{\frac{3(n-m)}{2}+\frac{(-1)^{n-m}-1}{4}}=2^{\frac{n}{2}-\frac{3+(-1)^n}{4}}.
    \end{aligned}
    \end{equation}

    \item[b)] $m$ is even. Oppositely, in this case, $b_{j,k}(\vec{x})=\sum_{\{v_i,v_j,v_k\}\in E}x_i=1$ only as $v_j\in A_x,\ v_k\in A_{-x}$ and vice versa, so the edge set $E^{(2)}_{\vec{x}}$ represented by the function $b_{j,k}(\vec{x})$ describes a \emph{complete bipartite graph}. Denote the cardinality of the following sets as  $w_{00}=|\{i|a_i=1,x_i=0,z_i=0\}|$, $w_{01}=|\{i|a_i=1,x_i=0,z_i=1\}|$, $w_{10}=|\{i|a_i=1,x_i=1,z_i=0\}|$, $w_{11}=|\{i|a_i=1,x_i=1,z_i=1\}|$. One can utilize the symmetry of the complete bipartite graph to demonstrate 
    \begin{equation}
    \begin{aligned}
        &\Tr{P_{\vec{x},\vec{z}}\kb{G}} \\
        =&2^{-n}\sum_{w_{01},w_{11}}\sum_{w_{00},w_{10}}{m_0 \choose w_{01}}{m_1 \choose w_{11}}{n-m-m_0 \choose w_{00}}{m-m_1 \choose w_{10}}(-1)^{w_{01}+w_{11}+(w_{01}+w_{00})(w_{11}+w_{10})}\\
	=&2^{-n}\sum_{w_{01},w_{11}}\left[{m_0 \choose w_{01}}{m_1 \choose w_{11}}(-1)^{w_{01}+w_{11}} \sum_{w_{00},w_{10}}{n-m-m_0 \choose w_{00}}{m-m_1 \choose w_{10}}(-1)^{(w_{01}+w_{00})(w_{11}+w_{10})}\right]\\
	=&2^{-n}\sum_{w_{01},w_{11}}{m_0 \choose w_{01}}{m_1 \choose w_{11}}(-1)^{w_{01}+w_{11}}\cdot\left\{\begin{array}{l}
				(-1)^{c_{01}c_{11}}(1+(-1)^{c_{01}})^{m-m_1},\quad m_0=n-m\\
				(-1)^{c_{01}c_{11}}(1+(-1)^{c_{10}})^{n-m-m_0},\quad m_1=m\\
				2^{n-m_0-m_1-1}, \quad m_0\neq n-m,\ m_1\neq m
				\end{array}\right.\\
	=&\left\{\begin{array}{l}
	        \frac{1}{2},\quad \text{when }m\neq0\&n,\ (m_0=0|(n-m))\&(m_1=0|m)\\
				1,\quad \text{when }\{m=0,\ m_0=0\}|\{m=n,\ m_1=0\}\\
				0,\quad o.w.\\
				\end{array}\right.
    \end{aligned}
    \end{equation}
When $m=0$ or $m=n$,
    \begin{equation}
    \begin{gathered}
	\sum_{\vec{z}}\left(\Tr{P_{\vec{x},\vec{z}}\kb{G}}\right)^4=1.\\
 	\sum_{\vec{z}}\left|\Tr{P_{\vec{x},\vec{z}}\kb{G}}\right|=1.
    \end{gathered}
    \end{equation}
When $m\neq 0$ and $m\neq n$,
    \begin{equation}
    \begin{gathered}
	\sum_{\vec{z}}\left(\Tr{P_{\vec{x},\vec{z}}\kb{G}}\right)^4=4\cdot\frac{1}{2^4}=\frac{1}{4}.\\
	\sum_{\vec{z}}\left|\Tr{P_{\vec{x},\vec{z}}\kb{G}}\right|=4\cdot\frac{1}{2}=2.
    \end{gathered}
    \end{equation}
\end{itemize}

With the former analysis, one can obtain the Pauli-Liouville moment
\begin{equation}
\begin{aligned}
    \mathbf{m}_{2}(\ket{G_{3\textrm{-com}}})=&2^{-n}\sum_{\vec{x}}\sum_{\vec{z}}\left(\Tr{P_{\vec{x},\vec{z}}\kb{G}}\right)^4\\
    =&2^{-n}\left(\sum_{m\ \text{odd}}\frac{1}{2^{n-\frac{3+(-1)^n}{2}}}+\sum_{m\ \text{even}}\frac{1}{4}+\frac{3}{4}+\frac{3}{4}\cdot\frac{(1+(-1)^{n})}{2}\right)\\
    =&\frac{1}{8}+\frac{7}{2^{n+\frac{3-(-1)^n}{2}}},
\end{aligned}
\end{equation}
and
\begin{equation}
\begin{aligned}
    \mathbf{m}_{1/2}(\ket{G_{3\textrm{-com}}})=&2^{-n}\sum_{\vec{x}}\sum_{\vec{z}}\left|\Tr{P_{\vec{x},\vec{z}}\kb{G}}\right|^4\\
    =&2^{-n}\left(\sum_{m\ \text{odd}}2^{\frac{n}{2}-\frac{3+(-1)^n}{4}}+\sum_{m\ \text{even}}2-1-1\cdot\frac{(1+(-1)^{n})}{2}\right)\\
    =&2^{\frac{2n-7-(-1)^n}{4}}+1-2^{-n+\frac{1+(-1)^n}{2}}.
\end{aligned}
\end{equation}
The corresponding $\mathbf{M}_{2}(\ket{G_{3\textrm{-com}}})$ and $\mathbf{M}_{1/2}(\ket{G_{3\textrm{-com}}})$ can be calculated directly.

\subsection{Proof of Proposition \ref{th:n-com}}
The structure of the induced hypergraphs are shown in Fig.~\ref{fig:indGsym} (b).
In order to calculate $\Tr{U(G_{\vec{x},\vec{z}})}$, consider each term $B(\vec{a})=\bra{\vec{a}}U(G_{\vec{x},\vec{z}})\ket{\vec{a}}$. For convenience, let $\vec{a}=(\vec{a}_1,\vec{a}_0)$ where $\vec{a}_1$ and $\vec{a}_0$ represent the qubit of vertices in $A_x$ and $A_{-x}$ respectively. It is not hard to see that as long as $\vec{a}_0\neq\vec{1}$, the edges in $E^{(2)}_{\vec{x}}$ will give no addition phase. 

When $A_x\neq\emptyset$, there are two different cases:
\begin{itemize}
    \item[a)] $m_0\neq 0$ or $m_1\neq 0$: Define $\vec{a}_{-i}=\vec{a}+\hat{i}$ where $\hat{i}$ is the unit vector corresponding to the vertex $v_i\in A_{z,-x}$ or $v_i\in A_{z,x}$. Then one only needs to consider the case when $\vec{a}_0=\vec{1}$ since otherwise $B(\vec{a})$ adding $B(\vec{a}_{-i})$ will vanish to $0$. With the same consideration, one can find that $B(\vec{a})$ adding $B(\vec{a}_{-i})$ won't vanish iff $\vec{a}=(\vec{0},\vec{1})$ and $\vec{a}=(\vec{1},\vec{1})$. They have the same phase when $m_1$ is even, and the opposite phase when $m_1$ is odd. Thus, 
    \begin{equation}
        \Tr{U(G_{\vec{x},\vec{z}})}=2\left(1+(-1)^{m_1}\right).
    \end{equation}
    \item[b)] $m_0=m_1=0$: In this case, $B(\vec{a})=1$ for nearly all the vectors $\vec{a}$. One can find that $B(\vec{a})=-1$ iff $\vec{a}=(\vec{0},\vec{1})$ and $\vec{a}=(\vec{1},\vec{1})$. Thus, 
    \begin{equation}
        \Tr{U(G_{\vec{x},\vec{z}})}=2^n-4.
    \end{equation}
\end{itemize}
When $A_x=\emptyset$, i.e., $m=0$, it is slightly different. In this case, since there are only $Z$ gates, $\Tr{U(G_{\vec{x},\vec{z}})}=0$ for all nonzero $m_0$, and $\Tr{U(G_{\vec{x},\vec{z}})}=2^n$ for $m_0=0$. As a result, the PL-moment with $\alpha=2$ shows
\begin{equation}\label{eq:ncardisum}
\begin{aligned}
    \mathbf{m}_{2}(\ket{G_{n\textrm{-com}}})=&\frac{1}{2^{5n}}\sum_{\vec{x}}\sum_{\vec{z}}\left(\Tr{U(G_{\vec{x},\vec{z}})}\right)^4\\
    =&\frac{1}{2^{5n}}\left(\left(2^n-1\right)\left(\left(\frac{2^n}{2}-1\right)\cdot 4^4+1\cdot\left(2^n-4\right)^4\right)+1\cdot1\cdot\left(2^n\right)^4\right)\\
    =&1-16\cdot2^{-n}+112\cdot2^{-2n}-224\cdot2^{-3n}+128\cdot2^{-4n},
\end{aligned}
\end{equation}
and the PL-moment with $\alpha=\frac{1}{2}$ shows
\begin{equation}\label{eq:ncardisum2}
\begin{aligned}
    \mathbf{m}_{1/2}(\ket{G_{n\textrm{-com}}})=&\frac{1}{2^{2n}}\sum_{\vec{x}}\sum_{\vec{z}}\left|\Tr{U(G_{\vec{x},\vec{z}})}\right|\\
    =&\frac{1}{2^{2n}}\left(\left(2^n-1\right)\left(\left(\frac{2^n}{2}-1\right)\cdot 4+1\cdot\left(2^n-4\right)\right)+1\cdot1\cdot\left(2^n\right)\right)\\
    =&3-10\cdot2^{-n}+8\cdot2^{-2n}.
\end{aligned}
\end{equation}
In the second line of Eq.~\eqref{eq:ncardisum} and Eq.~\eqref{eq:ncardisum2}, the term $\left(2^n-1\right)$ represents the case $A_x\neq\emptyset$, and the term $\left(\frac{2^n}{2}-1\right)$ represents the number of $\Tr{U(G_{\vec{x},\vec{z}})}=4$ when $m_0\neq 0$ or $m_1\neq 0$. The corresponding $\mathbf{M}_{2}(\ket{G_{3\textrm{-com}}})$ and $\mathbf{M}_{1/2}(\ket{G_{3\textrm{-com}}})$ can be calculated directly.

\end{appendix}

\end{document}